\documentclass[a4paper,UKenglish,cleveref, autoref, thm-restate]{lipics-v2019}

\title{Reachability in Dynamical Systems with Rounding}

\author{Christel Baier}{
Technische Universität Dresden, Germany	
}{%
}{https://orcid.org/0000-0002-5321-9343%
}{%
}

\author{Florian Funke}{
Technische Universität Dresden, Germany	
}{%
}{https://orcid.org/0000-0001-7301-1550%
}{%
}
\author{Simon Jantsch}{
Technische Universität Dresden, Germany	
}{%
}{https://orcid.org/0000-0003-1692-2408%
}{%
}

\author{Toghrul Karimov}{
Max Planck Institute for Software Systems, Saarland Informatics Campus, Saarbr\"ucken, Germany	
}{%
}{%
}{%
}

\author{Engel Lefaucheux}{
Max Planck Institute for Software Systems, Saarland Informatics Campus, Saarbr\"ucken, Germany	
}{%
}{%
}{%
}
\author{Jo\"el Ouaknine}{
Max Planck Institute for Software Systems, Saarland Informatics Campus, Saarbr\"ucken, Germany	\and Department of Computer Science, Oxford University, UK
}{%
}{https://orcid.org/0000-0003-0031-9356%
}{%
Supported by ERC grant AVS-ISS (648701).
}

\author{Amaury Pouly}{
Université de Paris, CNRS, IRIF, F-75006, Paris, France
}{%
}{https://orcid.org/0000-0002-2549-951X%
}{Supported by CODYS project ANR-18-CE40-0007.
}

\author{David Purser}{
Max Planck Institute for Software Systems, Saarland Informatics Campus, Saarbr\"ucken, Germany	
}{%
}{https://orcid.org/0000-0003-0394-1634%
}{%
}

\author{Markus A.\,Whiteland}{
Max Planck Institute for Software Systems, Saarland Informatics Campus, Saarbr\"ucken, Germany	
}{%
}{https://orcid.org/0000-0002-6006-9902%
}{%
}

\authorrunning{C.\,Baier et al.} %

\Copyright{Christel Baier, Florian Funke, Simon Jantsch, Toghrul Karimov, Engel Lefaucheux, Jo\"el Ouaknine, Amaury Pouly, David Purser and Markus A.\,Whiteland} %



\category{} %

\supplement{}%

\funding{This work was funded by DFG grant 389792660 as part of TRR~248 (see \url{https://perspicuous-computing.science}), the Cluster of Excellence EXC 2050/1 (CeTI, project ID 390696704, as part of Germany's Excellence Strategy), DFG-projects BA-1679/11-1 and BA-1679/12-1, and the Research Training Group QuantLA (GRK 1763).}%

\nolinenumbers %

\hideLIPIcs  %

\EventEditors{Nitin Saxena and Sunil Simon}
\EventNoEds{2}
\EventLongTitle{40th IARCS Annual Conference on Foundations of Software Technology and Theoretical Computer Science (FSTTCS 2020)}
\EventShortTitle{FSTTCS 2020}
\EventAcronym{FSTTCS}
\EventYear{2020}
\EventDate{December 14--18, 2020}
\EventLocation{BITS Pilani, K K Birla Goa Campus, Goa, India (Virtual Conference)}
\EventLogo{}
\SeriesVolume{182}
\ArticleNo{}

\hypersetup{bookmarksnumbered=true}

\newtheorem*{problem}{Problem}
\newtheorem{openproblem}[theorem]{Open Problem}

\usepackage[ruled]{algorithm2e}

\usepackage{stmaryrd}

\usepackage{wrapfig}
\usepackage{float}
\usepackage[bold]{complexity}

\usepackage[most]{tcolorbox}
\newcommand{\naturals}{\mathbb{N}}
\newcommand{\ints}{\mathbb{Z}}
\newcommand{\reals}{\mathbb{R}}
\newcommand{\rationals}{\mathbb{Q}}
\newcommand{\algs}{\mathbb{A}}

\newcommand{\round}[1]{[#1]}
\newcommand{\floor}[1]{\ensuremath{\left\lfloor #1 \right\rfloor}}
\newcommand{\ceil}[1]{\ensuremath{\left\lceil #1 \right\rceil}}

 \newcommand{\itx}[3]{\texttt{if } (#1) \texttt{ then } #2 \texttt{ else } #3}

 \newcommand{\itxss}[3]{\begin{aligned}[t]
    \texttt{if } (#1) &\texttt{ then } #2 \\ & \texttt{ else } #3
  \end{aligned}}

\newcommand{\abs}[1]{\left|#1\right|}

\newcommand{\sqbr}[1]{\ensuremath{[#1]}}
\newcommand{\abr}[1]{\left\langle#1\right\rangle}

\newcommand{\vecit}[2]{#1^{(#2)}}
 \newcommand{\vecitcomp}[3]{(\vecit{#1}{#2})_{#3}}
 \newcommand{\avecitcomp}[3]{\abs{\vecitcomp{#1}{#2}{#3}}}
 \newcommand{\mx}[1]{\max\left\{#1\right\}}

\newcommand{\dimension}{d}
\newcommand{\ck}{k}
\newcommand{\target}{y}
\newcommand{\numangles}{R}

\usepackage{xfrac}
\usepackage{diagbox}
\usepackage{multirow}

\usepackage{xcolor}

\newcommand\colortext[2]{{\color{#1}#2}}

\newcounter{SideNoteCounter} \stepcounter{SideNoteCounter}

\usepackage{xr}

\usepackage{pgf}
\usepackage{tikz}
\usetikzlibrary{arrows,automata}
\usetikzlibrary{calc}
\usetikzlibrary{arrows,decorations.markings}

        \definecolor{ApyOrange}{RGB}{230,159,  0}
        \definecolor{ApySkyBlue}{RGB}{ 86,180,233}
        \definecolor{ApyBluishGreen}{RGB}{  0,158,115}
        \definecolor{ApyYellow}{RGB}{240,228, 66}
        \definecolor{ApyBlue}{RGB}{  0,114,178}
        \definecolor{ApyVermillion}{RGB}{213, 94,  0}
        \definecolor{ApyReddishPurple}{RGB}{204,50,110}
        \definecolor{MRed}{RGB}{139,30,63}

\theoremstyle{remark}
\newtheorem{case}{Case}
\usepackage{etoolbox}
\AtBeginEnvironment{proof}{\setcounter{case}{0}}
\begin{document}

\maketitle

\begin{abstract}
We consider reachability in dynamical systems with discrete linear updates, but with fixed digital precision, i.e., such that values of the system are rounded at each step.
Given a matrix $M \in \mathbb{Q}^{\dimension \times \dimension}$, an initial vector $x\in\mathbb{Q}^{d}$, a granularity $g\in \rationals_+$ and a rounding operation $[\cdot]$ projecting a vector of $\mathbb{Q}^{d}$ onto another vector whose every entry is a multiple of $g$,
we are interested in the behaviour of the orbit $\mathcal{O}=\abr{[x], [M[x]],[M[M[x]]],\dots}$, i.e., the trajectory of a linear dynamical system 
in which the state is rounded after each step.
For arbitrary rounding functions with bounded effect, we show that the complexity of deciding point-to-point reachability---whether a given target $y \in\mathbb{Q}^{d}$ belongs to $\mathcal{O}$---is $\PSPACE$-complete for hyperbolic systems (when no eigenvalue of $M$ has modulus one). We also establish decidability without any restrictions on eigenvalues for several natural classes of rounding functions.
\end{abstract}

\newpage
\section{Introduction}

A \emph{discrete-time linear dynamical system} in ambient space $\mathbb{Q}^d$ is specified via a linear transformation together with a starting point. The state of the system is then updated at each step by applying the linear transformation, giving rise to an \emph{orbit} (or infinite trajectory) in $\mathbb{Q}^d$. 

One of the most well-known questions for such systems is the \emph{Skolem Problem}, which asks whether the orbit ever hits a given $(d-1)$-dimensional hyperplane.\footnote{The Skolem Problem is usually formulated in terms of linear recurrence sequences, but is equivalent to the description given here.} This problem has long eluded decidability, although instances of dimension $d \leq 4$ are known to be solvable (see, e.g., the survey~\cite{OW15}). Another natural problem is \emph{point-to-point reachability}\footnote{Historically this problem has been known as the \textit{orbit problem}, however there are now multiple `orbit problems' (polytope  reachability, hyperplane reachability, (semi-)algebraic set reachability,... etc.) and so we specify point-to-point reachability.}, known to be decidable in polynomial time~\cite{KannanL86}. In both cases, however, one assumes \emph{arbitrary} precision, which arguably is unrealistic for simulations carried out on digital computers. In this paper, we therefore turn our attention to instances of these problems in which the numerical state of the system is rounded to finite precision at each time step. This leads us to the following definition:

\begin{problem}[Rounded Point-to-Point Reachability (Rounded P2P)]\label{prob:kl}
Given a matrix $M \in \mathbb{Q}^{\dimension \times \dimension}$, an initial vector $x\in\mathbb{Q}^{d}$, a target vector $\target\in\mathbb{Q}^{d}$, a granularity $g\in \rationals_+$, and a rounding operation $[\cdot]$ projecting a vector of $\mathbb{Q}^{d}$ onto another vector whose every entry is a multiple of $g$, let the orbit $\mathcal{O}$ of this system
be the infinite sequence $\abr{[x], [M[x]],[M[M[x]]],\dots}$, i.e., $\vecit{x}{0} = [x]$ and $\vecit{x}{i+1} = [M\vecit{x}{i}]$.  The \textbf{Rounded Point-to-Point Reachability (Rounded P2P) Problem} asks whether $[\target] \in \mathcal{O}$. \end{problem}

\subparagraph*{Main contributions.}
We make the following contributions, summarised in \cref{fig:results}:

\begin{enumerate}
	\item We introduce a family of natural problems, Rounded P2P (parameterised by the rounding function),
which to the best of our knowledge has not previously been studied.

	\item We show that for hyperbolic systems (i.e., those whose associated linear transformation has no eigenvalue of modulus 1) the Rounded P2P Problem is solvable---and is in fact $\PSPACE$-complete---for any `reasonable' (i.e., bounded-effect) rounding function. It is interesting to note, in contrast, that exact P2P reachability is known to be solvable in polynomial time. Our approach to solving the Rounded P2P Problem relies on the observation that, outside a ball of exponential size, the change in magnitude of the system state at each step dwarfs any effect due to rounding. It thus suffices to exhaustively examine the effect of the dynamics inside an exponentially bounded state space.

	\item In the general case (without any restriction on the magnitude of eigenvalues), the effect of rounding may forever remain non-negligible, requiring a careful analysis. We have not been able to solve the problem in full generality, but we do provide a complete solution for certain natural classes of rounding functions. More precisely, assume that the linear transformation has been converted to Jordan normal form (now requiring us to work with complex algebraic numbers). We can then solve the Rounded P2P Problem under two natural classes of rounding functions:\begin{enumerate}[(a)]
	\item \emph{Polar rounding functions:} given a complex number of the form $Ae^{i\theta}$, such functions round $A$ and $\theta$ independently. In such instances we can handle in $\EXPSPACE$ all reasonable rounding functions on $A$, and what we view as the only natural rounding function on $\theta$. 

	\item \emph{Argand rounding:} given a complex number of the form $a+bi$, the \emph{Argand truncation} will round $a$ and $b$ independently downwards (in magnitude), ensuring that the modulus 
never increases. Similarly, the \emph{Argand expansion} (which rounds $a$ and $b$ independently upwards) guarantees that the modulus can only increase. Under such rounding functions, we show decidability in $\EXPSPACE$.
	\end{enumerate}

	\item We highlight some limitations of our methods, identifying a simple but technically challenging open problem, which points to some of the key difficulties in solving the Rounded P2P Problem in full generality. More precisely, we consider minimal error rounding for a simple rotation in two-dimensional space, for which Rounded P2P is presently open.

\end{enumerate}

\begin{remark*}
  It is worth noting that the rounded versions of the Skolem Problem (does the rounded orbit ever hit a $(d-1)$-dimensional hyperplane?)\ and the Positivity Problem (does the rounded orbit ever hit a $d$-dimensional half-space?)\ remain at least as hard as their exact integer counterparts, since over the integers rounding has no effect; the decidability of these problems therefore remains open. However, the rounded versions of reaching a bounded polytope or a bounded semialgebraic set (problems not known to be decidable in the exact setting~\cite{COW15,AO019}) reduce to a finite number of Rounded P2P reachability queries (since a bounded set can contain only finitely many rounded points). These observations together motivate our focus, in the present paper, on the Rounded P2P Problem.
\end{remark*}

\begin{figure}
\centering
\begin{tabular}{|p{4.45cm}|p{2.1cm}|p{3.7cm}|c|}
\hline
\multirow{2}{*}{Rounding type}  &\multirow{2}{2.1cm}{Hyperbolic Systems}  & \multicolumn{2}{c|}{No restrictions on eigenvalues } \\ \cline{3-4}
&&  \parbox[t]{3cm}{Jordan normal form\\ (Note: no hardness)} & General \\\hline
Polar $Ae^{i\theta}$ & \multirow{4}{2.1cm}{\\[0.3cm]\centering{$\PSPACE$-complete}, \cref{sec:non-mod-1}} & $\EXPSPACE$, \cref{sec:polar} & \multirow{3}{1.8cm}{\\[0cm]\centering Open but $\PSPACE$-hard} \\\cline{1-1}\cline{3-3}
\raggedright Argand truncation or expansion &  & $\EXPSPACE$, \cref{sec:truncated} &  \\\cline{1-1}\cline{3-3}
Argand minimal error &  & \raggedright Open (difficulties highlighted in \cref{sec:problemswith2x2}) &  \\\cline{1-1}\cline{3-3}
Arbitrary bounded-effect &  & \multicolumn{1}{l}{Open (\cref{openproblem:klabr})} & \diagbox[width=2.2cm]{$\mathstrut$}{$\mathstrut$} \\\hline
\end{tabular}
\caption{Decidability and complexity table for the Rounded P2P Problem.}
\label{fig:results}
\end{figure} 

It is interesting to consider rounded reachability problems in the stochastic setting, i.e., Markov chains. One observes that the state space $\mathcal{[S]} = \{[x] \in [0,1]^d \ | \ x \text{ sub-stochastic}\}$ is finite, which entails decidability of virtually any reachability problem, including Skolem and Positivity. This is somewhat arresting, since without rounding reachability problems are known to be exactly as hard for stochastic systems as for general systems~\cite{AkshayAOW2015}. In any event, one should note that ensuring that for all $x \in \mathcal{[S]}$, $[Mx] \in \mathcal{[S]}$ requires some care, as arbitrary rounding does not necessarily preserve (sub-)stochasticity.

\subsection*{Related work}

With the emerging use of numerical computations during the 80s, doubts were raised concerning the transferability of results about dynamical systems obtained by simulation in finite-state machines. In this direction, the sensitivity that a rounding function may have on the long-term behaviour of a dynamical system is studied in \cite{BeckR1987}. How rounded orbits can be simulated by actual orbits of the dynamical system is investigated in \cite{HammelYG1988, NusseN1988}. 

The series of papers \cite{Blank1984, Blank1989b, Blank1989, Blank1994} examines which statistical properties of a discrete dynamical system are preserved under the introduction of a rounding function, a good summary of which can be found in Blank's book \cite[Chapter 5]{Blank1997}. As the rounding is refined, some properties of the discretized orbits follow probabilistic laws asymptotically, as shown in \cite{DiamondV1998, DiamandV2002}.
The paper~\cite{DiasLC2011} studies how volatile statistical notions are in the presence of finite precision (such as the mean distance of two orbits of discrete dynamical systems). 

Another line of research focuses on discretized rotations in $\ints^2$ and higher-dimensional lattices \cite{LowensteinHV1997, AkiyamaBBPT2005}. A connection from roundoff problems in the $2$-dimensional case to expanding maps on the $p$-adic integers is described in \cite{BosioV1999, VivaldiV2003}. Building on this, \cite{Vivaldi2006} conjectures periodicity of all orbits of these discretized rotations in $\ints^2$. It is shown in \cite{AkiyamaP2013} that there are infinitely many periodic orbits, and \cite{PethoTW2019} attempts to concisely describe points leading to periodic orbits.

In the context of model checking, continuous dynamical systems have been translated into discrete models, mainly timed automata that approximate the behaviour of the original system \cite{MB2008, CarterN2012, SchivoL2017}. On a more general level, one can observe a growing interest in the systematic study of roundoff errors inherent in finite precision computations \cite{GoubaultP2011, SolovyevJRG2015, IzychevaD2017, MagronCD2017, MoscatoTDM2017, DarulovaINRBB2018}.

\section{Rounding functions}

Let $\naturals{},\ints,\rationals{},\reals{},\algs{}$ be the naturals, integers, rationals, reals, and algebraic numbers respectively.

\paragraph*{Rounding real numbers}
Let $g \in \reals_+$ be a granularity. We define our rounding functions taking values to integers, i.e., $g = 1$. For $g\ne 1$ we consider $\round{x} = g\cdot \round{x/g}$.  Given a set $S$, we let $[S] = \{[x] \mid x\in S\}.$

The floor  function $\floor{x}$ and ceiling functions $\ceil{x}$ are well-known rounding functions in mathematics and computer science. We recall two further rounding functions:
\begin{itemize}
\item \emph{Minimal error rounding} rounds to the \textit{nearest} value: $[x] = \arg\min_{y\in\mathbb{Z}} \abs{x-y}$. If $\abs{x-y} =0.5$ an arbitrary but deterministic choice must be made (e.g. to round up).
\item \emph{Truncation} (`towards zero rounding', to cut off the remaining bits): if $x > 0$ then $\floor{x}$ else $\ceil{x}$, or \emph{expansion}: if $x > 0$ then $\ceil{x}$ else $\floor{x}$. 
\end{itemize}
Whenever possible, we prefer to analyse the problems without choosing a specific rounding function, relying only upon the property of \emph{bounded effect}:
\begin{definition}\label{defn:bounded:real}
A real rounding function $[\cdot]\colon \reals\to\reals$ has bounded effect if there exists $\Delta$ such that $\abs{x - [x]} \le \Delta$ for all $x$.
\end{definition}

\paragraph*{Rounding complex numbers}
Complex numbers have both a real and imaginary part. Thus one can consider rounding each of the components separately, which we call \textit{Argand rounding}. Consider $x = a+bi$ with $a,b\in\reals{}$, then let $[x] = [a] + [b]i$, where $[\cdot]$ can be any real rounding function (leading to \emph{Argand truncation}, \emph{Argand expansion} and \emph{Argand minimal error rounding} functions).

However, complex numbers can also be readily represented using polar coordinates as follows:
a number is represented as $x = Ae^{i\theta}$, where $A$ is the modulus and $\theta$ is the angle between the 2-d coordinates $(1,0)$ and $(a,b)$ (when represented as $a+bi$).
Then, a \emph{polar rounding} function rounds $A$ and $\theta $ independently, i.e. $[x] = [A] e^{i[\theta]}$. The rounding of $[A]$ can be any real rounding function. For the rounding of the angle we always assume minimal error rounding. That is,  given granularity $\theta_g = \frac{\pi}{\numangles}$ for some $\numangles \in \naturals{}$, then $[\theta]$ is a multiple of $\theta_g$ with minimal error and arbitrary but deterministic tie breaking.

We generalise non-specific bounded-effect rounding to the complex numbers.
\begin{definition}\label{defn:bounded:complex}
A complex rounding function $[\cdot]\colon \mathbb{C}\to\mathbb{C}$ has bounded effect \textbf{on the modulus} if there exists $\Delta$ such that $\abs{\abs{x} - \abs{[x]}} \le \Delta$ for all $x$.
\end{definition}

Argand and polar roundings are both defined by applying  bounded-effect real rounding functions to each component, and have bounded effect under \cref{defn:bounded:complex}.  However, note the distinction with \cref{defn:bounded:real}; polar rounding can exhibit arbitrary large effects (in the following sense: given any $\Delta > 0$, one can always find $x \in \mathbb{C}$ such that 
$\abs{x - [x] } > \Delta$), but nevertheless has only bounded effect \textit{on the modulus}.

\begin{definition}[\sqbr{K}-Ball]
  \label{def:kball}
Given a complex rounding function $[\cdot]$ and an integer $K$ let a $[K]$-ball be the set of admissible points of modulus at most $K$, i.e., $\{[x] \mid x\in \mathbb{C}, \ \abs{[x]} \le K\}$.\end{definition} 

\paragraph*{Rounding vectors}In general, a rounding function on $\mathbb{K}$ induces a rounding function on vectors $\mathbb{K}^d$, where $[(x_1,\dots,x_d)]= ([x_1],\dots,[x_d])$, although not all rounding functions on vectors need take this form. We generalise non-specific bounded-effect rounding to vectors.

\begin{definition}\label{defn:bounded:vector}
A rounding function $[\cdot]\colon \mathbb{K}^d\to\mathbb{K}^d$ has bounded effect on the modulus if there exists $\Delta$ such that $\abs{\abs{x}_k - \abs{[x]}_k} \le \Delta$ for all $x$ and every $k\in \{1,2,\dots, d\}$.
\end{definition}

Finally, we assume that all of our rounding functions can be computed in polynomial time and are fixed (rather than inputs) in our problems, and thus $\Delta$ is also a fixed parameter.

\section{Hyperbolic systems}\label{sec:non-mod-1}

In this section we establish our first main result for hyperbolic systems, which we first define:

\begin{definition}[Hyperbolic System~{\cite[Section~1.2]{katok_hasselblatt_1995}}] A linear map represented by the matrix $M \in \mathbb{R}^{\dimension \times \dimension}$ is hyperbolic if all of its eigenvalues have modulus different from one.
\end{definition}

\begin{theorem}\label{thm:all-not-1}
The Rounded P2P Problem is $\PSPACE$-complete for hyperbolic linear maps represented by rational matrices and real rounding functions with bounded effect.
\end{theorem}

We first demonstrate that the problem is in $\PSPACE$ for matrices in Jordan normal form, to which we will reduce the general case in a second step. As the passage to Jordan normal form inevitably introduces complex numbers, $\PSPACE$ membership will be shown for Jordan normal form matrices over the algebraic numbers and, accordingly, complex rounding functions with bounded effect on the modulus. To complete the picture we show hardness for hyperbolic systems (in fact, the hardness result applies even for non-hyperbolic systems, that is for matrices whose eigenvalues may include 1). 

\subsection{Membership in \texorpdfstring{$\PSPACE$}{PSPACE}}
We now prove the membership part of \Cref{thm:all-not-1} under the additional assumption that the matrices are in Jordan normal form. %

\begin{lemma}\label{thm:jnf-bounded-rounding}
The Rounded P2P Problem decidable in $\PSPACE$ for any complex rounding function with bounded effect on the modulus $\Delta$ and hyperbolic matrices $M \in \algs^{\dimension \times \dimension}$ in Jordan normal form.
\end{lemma}

\begin{proof}
We consider a single Jordan block of dimension $\dimension$ with eigenvalue $\lambda$. If the matrix $M$ has multiple Jordan blocks, the algorithm can be run in lock step\footnote{By running processes in lock step, here and elsewhere, we mean running all of the processes simultaneously (interleaving instructions for each process) until either $\vecit{x}{i} = \target{}$ or one of the processes concludes non-reachability.} for each block. Hence, without loss of generality we let \[ M = 
\left[\begin{smallmatrix}
  \lambda & 1  
\\ & \lambda & 1  
\\ & & \ddots & 1
\\& && \lambda
\end{smallmatrix}\right].
\]

The idea will be to show that for $\abs{\lambda} > 1$, for values large enough growth will outstrip the rounding, and the orbit will grow beyond the target, never to return. If $\abs{\lambda} < 1$ and the orbit gets large enough, it will begin to contract again, so we choose a ball large enough to contain the whole orbit. We do not consider the case $\abs{\lambda} = 1$ here.

Formally, in each dimension $\ck \in \{1,\dots,\dimension\}$ we compute a radius $C_{\ck}$, defining a $[C_{\ck}]$-ball of radius $C_{\ck}$ about $0$, containing $x_\ck$ and $y_\ck$ such that for all $z$ in the orbit $\mathcal{O}$ if $z_\ck \not\in [C_{\ck}]$-ball then $[Mz]_\ck \not\in [C_{\ck}]$-ball. That is, if the orbit has left the ball, it will never come back. The algorithm proceeds by simulating the orbit from $x$ until one of the following occurs.
\begin{itemize}
	\item $\target$ is found, in which case $\textsc{return yes}$, or
	\item a point repeats, in which case $\textsc{return no}$, or
	\item a point $\vecit{x}{i}$ is found such that $\abs{\vecitcomp{x}{i}{k}} \ge C_{\ck}$ for some $k$, in which case $\textsc{return no}$.
\end{itemize}
Since $B = [\{x \in \mathbb{R}^\dimension{} \mid \text{for all } k \ \abs{x_k} \le C_k\}]$ is finite, one of the three must occur. Remembering all previous points would require too much space.
Therefore we record a counter of the number of steps taken and once this exceeds the maximum number of points then we know some point must have been repeated (possibly many times by this point).
Let $C = \max_i C_k$, then the bounding hyper-cube of $B$ has $(2C/g)^d$ points, hence $B$ has fewer points. We show this number has at most exponential size in the description length of the input, and hence can be represented in $\PSPACE$.

\begin{case}[suppose $\abs{\lambda} > 1$]
For the $\dimension$th component we have $\vecitcomp{x}{i+1}{\dimension} = [\lambda \vecitcomp{x}{i}{\dimension}]$. There is a bounded effect of the rounding $\Delta$, ensuring
$\avecitcomp{x}{i+1}{\dimension} \ge  \abs{\lambda} \avecitcomp{x}{i}{\dimension} - \Delta$.
So when $\abs{\lambda} \abs{\vecitcomp{x}{i}{\dimension}} - \Delta > \avecitcomp{x}{i}{\dimension}$, this component must grow. Let $\ell = \mx{1, \Delta, \abs{y_1},\ldots,\abs{y_d}}$. We define the radius
$C_\dimension := \frac{\Delta}{\abs{\lambda} - 1} + \ell$,
which satisfies the desired property described above.

Now suppose that the radius $C_k$ is defined so that
$C_k \leq \ell \sum_{j=0}^{d-k+1}(\frac{2}{|\lambda|-1})^{j}$ (holds for $k=d$) and assume that
$\avecitcomp{x}{i}{j} \le C_j$ for each $j \in \{k,\ldots,d\}$.
For the $\ck{-}1$th dimension the update is of the form $\vecitcomp{x}{i+1}{\ck{-}1} = [\lambda \vecitcomp{x}{i}{\ck{-}1} + 1 \vecitcomp{x}{i}{\ck}]$.
Since $\avecitcomp{x}{i}{\ck} \le C_\ck$, we have $\avecitcomp{x}{i+1}{\ck{-}1} \ge  \abs{\lambda} \avecitcomp{x}{i}{\ck{-}1} - \Delta - C_\ck$, and there is growth when
$\abs{\lambda} \avecitcomp{x}{i}{\ck{-}1} - \Delta - C_\ck > \avecitcomp{x}{i}{\ck{-}1}$,
i.e., when $\avecitcomp{x}{i}{\ck{-}1} > \frac{\Delta+C_\ck}{\abs{\lambda} -1}$. So,
we may define $C_{\ck{-}1} := \frac{\Delta+C_\ck}{\abs{\lambda} -1} + \ell$, which satisfies
the property described above, and moreover,
$C_{\ck{-}1} \leq \frac{2 C_{\ck}}{|\lambda|-1} + \ell \leq \ell \sum_{j = 0}^{d - (k-1)+1} (\frac{2}{|\lambda| - 1})^j $ due to our choice of $\ell$. Repeat for all remaining components $k-2,\dots, 1$.

Now $C_\ck \leq \ell \sum_{j=0}^{\dimension}(\frac{2}{|\lambda|-1})^j \leq  \ell (d+1)(1 + (\frac{2}{|\lambda|-1})^{d})$ for each $\ck$, and the claim follows.
\end{case}

\begin{case}[suppose $\abs{\lambda} < 1$] 
We require the ball to have the property that if the orbit leaves, it will never come back. However for $\abs{\lambda} < 1$, while initially there may be some growth (due to other components), once large enough $\abs{\lambda}$ will dominate and the modulus will decrease. Therefore, we want to ensure we choose the ball large enough that the orbit will never leave the ball in the first place. The following definitions of the radii $C_j$ can easily be altered to furnish this requirement.

Consider the last component $d$: we have $\abs{\vecitcomp{x}{i+1}{\dimension}} \leq \abs{\lambda }\abs{\vecitcomp{x}{i}{\dimension}} + \Delta$. Set again
$\ell = \max\{1,\Delta, \abs{y_1},\ldots, \abs{y_d}\}$ and define
$C_d := \frac{\Delta}{1-\abs{\lambda}} + \ell$; if $\avecitcomp{x}{i}{\dimension} \leq C_d$,
then $\avecitcomp{x}{i+1}{\dimension} \leq C_d$.

Having fixed $C_{\ck'}$ for $\ck'\in \{\ck,\dots, \dimension\}$, consider component $\ck -1$: We have $\vecitcomp{x}{i+1}{\ck{-}1} = [\lambda \vecitcomp{x}{i}{\ck{-}1} + \vecitcomp{x}{i}{\ck}]$, and so $\abs{\vecitcomp{x}{i+1}{\ck{-}1}} \le \abs{\lambda} \abs{\vecitcomp{x}{i}{\ck{-}1}} + \abs{\vecitcomp{x}{i}{\ck}} + \Delta$.
Let us define $C_{k-1}:= \frac{C_\ck + \Delta}{1-\abs{\lambda}} + \ell$.
Now if $\avecitcomp{x}{i}{\ck{-}1} \le C_{\ck{-}1}$
then $\avecitcomp{x}{i+1}{\ck{-}1} \le C_{\ck{-}1}$. Repeat for each remaining component. It
can be shown, similar to the previous case, that $C_\ck \leq \ell (d+1)(1 + (\frac{2}{1-|\lambda|})^{d})$ for each $\ck$, and this concludes the proof.
\qedhere\end{case}
\end{proof}

\paragraph*{Reducing the general form to Jordan normal form}
In the previous section we assumed that the matrix is always in Jordan normal form, which is a significant restriction. In this section we will not assume Jordan normal form, which means we cannot make any assumption about the rounding, other than being of bounded effect, to prove  \cref{thm:all-not-1}. After a change of basis properties such as `rounding towards zero' may not be preserved.
\begin{proof}[Proof (upper bound of \cref{thm:all-not-1})]
Let $\Delta$ be the fixed, bounded effect on the modulus of $[\cdot]$. Let us consider hyperbolic $M= P J P^{-1}\in \mathbb{Q}^{\dimension \times \dimension}$.
We ask whether $\vecit{x}{i+1} = \target$ for some $i$.
Observe that $\vecit{x}{i+1}=[M\vecit{x}{i}]=M\vecit{x}{i}+e(M\vecit{x}{i})$
where $e(x):=[x]-x\in\left[-\Delta,\Delta\right]^d$ for any $x$ since $[\cdot]$
has bounded effect. Now if we define $\vecit{z}{i}:=P^{-1}\vecit{x}{i}$ we have that
\[
    \vecit{z}{i+1}=P^{-1}\vecit{x}{i+1}
        =P^{-1}(M\vecit{x}{i}+e(M\vecit{x}{i}))
        =J\vecit{z}{i}+P^{-1}e(PJ\vecit{z}{i})
        =\llparenthesis J\vecit{z}{i}\rrparenthesis
\]
where $\llparenthesis z\rrparenthesis:=z+P^{-1}e(Pz)$ for any $z$. The question $\vecit{x}{i}\stackrel{?}{=}y$ for some $i$
now becomes equivalent to $\vecit{z}{i}\stackrel{?}{=}P^{-1}y$. But note that the system for $\vecit{z}{i}$
is in Jordan normal form and the rounding function $\llparenthesis\cdot\rrparenthesis$ has bounded effect on the modulus,
with bound $\Delta' \le \max_{1 \le k \le d} \max_{e \in \left[-\Delta,\Delta\right]^d } ( P^{-1}e)_\ck$. Since $\Delta$ is fixed and $P^{-1}$ is computable in polynomial time~\cite{DBLP:journals/ijfcs/Cai94}, then $\Delta'$ is of polynomial size.
Hence, we have produced in polynomial time an instance of the Rounded P2P problem with a matrix in Jordan normal form.
As the proof of \cref{thm:jnf-bounded-rounding} shows that this problem is solvable in $\PSPACE$ even if $\Delta$ is given as input, we can conclude that the $\PSPACE$ upper bound holds also for the general case. %
\end{proof}

\subsection{PSPACE-hardness}\label{sec:pspacehard}

We will prove $\PSPACE$-hardness (i.e., the lower bound of \cref{thm:all-not-1}) by reduction from quantified boolean formula (QBF), which is $\PSPACE$-complete~\cite{StockmeyerM73}. 
We do this by first encoding a simple programming language into the rounded P2P Problem. Then, we show that reachability in this language can solve QBF. Whilst a direct reduction is possible, we provide exposition via the language for two reasons; first, we will show that the language is robust to choice of rounding function (\cref{remark:choiceoffn}), and secondly the reduction results in an instance where all eigenvalues have modulus $1$, but by a small perturbation, we observe that the problem remains hard when all of the eigenvalues do not have modulus $1$ (\cref{remark:purtubation}).

The language will consist of $m$ instructions, operating over $d$ variables. Each instruction is a boolean map $f_i : [0,1]^d\to [0,1]^d$, where each dimension $i$ is updated using a logical formula of the $d$ inputs. Each of the $m$ instructions is conducted in turn and updating the $d$ variables is \textit{simultaneous} in each step. Thus, references to variable in a function are the evaluation in the previous step.  Once the $m$ instructions are complete, the system returns to the first instruction and repeats ($\vecit{x}{i} = (f_m \circ f_{m-1}\circ\dots\circ f_2\circ f_1)(\vecit{x}{i -1})$, see also~\cref{lst:1}).

An instruction is encoded into the rounded dynamical system using  a map $f_i : \mathbb{N}^d \to \mathbb{N}^d$ for $0 \le i \le m-1$, where instructions are of the form $(f_i(x))_j = \floor{(p_j \cdot x)}$ where $p_j$ in $\mathbb{Q}^d$. We demonstrate how to encode the required logical operations in a rounded dynamical system: and ($x_i \leftarrow x_j\wedge x_\ck = \floor{\frac{1 + x_j + x_\ck}{3} }$), or ($x_i \leftarrow x_j\vee x_\ck = \floor{\frac{1 + x_j + x_\ck}{2} }$), negation
 ($x_i \leftarrow \neg x_j = \floor{1-x_j}$), resetting a variable to \texttt{false} ($x_i \leftarrow \floor{0} $), copying a variable without change ($x_i \leftarrow \floor{x_i}$) or moving/duplicating a variable ($x_i \leftarrow \floor{x_j} $). To enable this, we will assume there is always access to the constant $1$ (or \texttt{true}) by an implicit dimension, fixed to $1$.

In multiple steps any logical formula can be evaluated. This can be done with auxiliary variables to store partial computations, where the instructions will in fact be multi-step instructions making use of a finite collection of auxiliary variables which will not be referenced explicitly. Meanwhile any unused variables can be copied without change.  In particular the syntax  $x_1 \leftarrow \itx{x_2}{x_3}{x_4}$ can be encoded, by equivalence with the logical formula $x_1 \leftarrow(( x_2\implies x_3) \wedge (\neg x_2 \implies x_4))$.

We ask, given some initial configuration $\vecit{x}{0}$, and a target $\target$: does there exist $i$ such that $\vecit{x}{i} =\target $. \label{sec:explode-dim} If there was just one step function, the system dynamics would be a direct instance of the rounded orbit semantics. When there are $m$ functions, we remark the sequence of functions can be encoded by taking $m$ copies of each variable, and each function $f_i$, can transfer the function from one copy to the next, zeroing the previous set of variables. That is, let \[M=
\left[\begin{smallmatrix}
  0 && &&f_m  
\\  f_1& 0
\\ & f_2 & 0
\\ && & \ddots & 
\\ && &f_{m-1}& 0
\end{smallmatrix}\right].
\]
Then the initial configuration becomes $(\vecit{x}{0},0,\dots,0)$, and the target becomes $(\target,0,\dots,0)$.

An abstraction of the language is depicted in \cref{lst:1}. It remains to show that QBF can be encoded in the language. 

\begin{algorithm}[t]
\KwIn{$x \in [0,1]^d $ initial vector, $y \in [0,1]^d $ target vector }

\While{$x \ne y$} {
	$x \leftarrow f_1(x)$ 

	$x \leftarrow f_2(x) \ \colortext{lipicsGray}{\stackrel{e.g.}{=}\  \begin{cases} x_1 \leftarrow x_2 \vee (x_5 \wedge x_3) \\ x_2 \leftarrow \itx{x_1\vee x_3}{x_6}{x_2} \\ x_3\leftarrow \mathtt{true}\\ \vdots \\ x_d \leftarrow x_4 \end{cases}}$

	$\vdots$

	$x \leftarrow f_m(x)$ 
}
\caption{System behaviour of the language}
\label{lst:1}
\end{algorithm}

\begin{lemma} \label{lemma:lang:qbf}
Reachability in this language can solve QBF.
\end{lemma}

\begin{proof}
Formally we write a program in our language to decide the truth of a formula of the form $\forall x_1 \exists x_2 \forall x_3 \dots \exists x_n \psi(x_1,\dots,x_n)$, where $\psi$ is a quantifier free boolean formula. For convenience we assume it starts with $\forall$, ends with $\exists$ and alternates. Formulae not in this form can be padded if necessary with variables which do not occur in the formula $\psi$.

The program will have the following variables: 
$x_1,\dots,x_n,\hat{\psi}, s^0_1,\dots,s^0_n, s^1_1,\dots,s^1_n$ and $c_1,\dots,c_{n}$. The bits $x_1,\dots,x_n$ represent the current allocation to the corresponding bit variables of $\psi$, and $\hat{\psi}$ will store the current evaluation of $\psi(x_1,\dots,x_n)$. To cycle through all allocations to $x_1,\dots,x_n$, the variables will be treated as a binary number and incremented by one many times, for this purpose the bits $c_1,\dots,c_{n}$ represent the carry bits when incrementing $x_1,\dots,x_n$.

The intuition of $s^z_i$ is the following: for fixed $x_1,\dots,x_{i-1}$ it stores the evaluation of  $ Q x_{i+1}\ Q' x_{i+2} \dots \exists x_n \psi(x_1,\dots,x_{i-1}, z, x_{i+1},\dots,x_n)$ where $Q,Q'\in \{\exists,\forall\}$ as required by the formula. Therefore the overall formula is true if and only if $s^0_1\wedge s^1_1$ is eventually true.

We define $3+n$ instructions, and each run through $f_1 \to f_{3+n}$ will cover exactly one allocation to $x_1,\dots,x_n$, with the next run through covering the next allocation that one gets by incrementing the rightmost bit. Once $x_{i+1}$ has been in both the $1$ state and the $0$ state for all values below, we have enough information to set  $s^{x_{i}}_{i+1}$. This is set when the carry-bit $c_{i+1}$ is one, which indicates that $x_{i+1}$ has visited both $0$ and $1$ and is being returned back to $0$ (thus setting $x_{i+1} = \dots = x_{n}$ back to $0$).

We let the initial configuration be $(0,0\dots,0)$. Note that this is hiding the implicit dimension that is always $1$. Each of the following step functions should be interpreted as copying any variable that is not explicitly set.

\noindent\begin{tabular}{@{}lll@{}}
\textbf{Step $1$.}
&\textbf{Step $2$.}
&\textbf{Step $3$.}
\\
Evaluate $\psi$ 
& Update either $s_n^0$ or $s_n^1$
& Start incrementing $x_{n}$ 
\\
$f_1(\cdot) = \begin{cases}\hat{\psi} \leftarrow  \psi(x_1,\dots,x_n)\end{cases}$
& 

$
f_2(\cdot) = \begin{cases}s_n^0 \leftarrow \itxss{x_n = 0}{\hat{\psi}}{s_n^0}\\
s_n^1 \leftarrow \itxss{x_n = 1}{\hat{\psi}}{s_n^1}\end{cases}$ 
&$f_3(\cdot) = \begin{cases}
x_{n} \leftarrow \neg x_{n}\\
c_{n} \leftarrow x_{n}\end{cases}$
\end{tabular}

\noindent\textbf{Step $3+n-i$, for  $i = n-1$ to $1$.}\\ 
If there is a carry, update $s_{i}^z$ and continue incrementing  \\
\noindent\begin{tabular}{ll}
\textbf{$i$ even ($x_i$ universally quantified):} & \textbf{$i$ odd  ($x_i$ existentially quantified):} \\
$f_{3+n-i}(\cdot)=$ & $f_{3+n-i}(\cdot)=$\\
 $\begin{cases}
x_{i} \leftarrow \itx{c_{i+1}}{\neg x_{i}}{x_{i}}\\
c_{i} \leftarrow c_{i+1} \wedge x_{i}\\
c_{i+1} \leftarrow 0\\
s_{i}^0 \leftarrow \itxss{c_{i+1} \wedge \neg x_i }{
	s_{i+1}^0 \wedge s_{i+1}^1
}{s_{i}^0}\\
s_{i}^1 \leftarrow \itxss{c_{i+1} \wedge x_i }{
	s_{i+1}^0 \wedge s_{i+1}^1
}{s_{i}^1}
\end{cases}$&
$\begin{cases}
x_{i} \leftarrow \itx{c_{i+1}}{\neg x_{i}}{x_{i}}\\
c_{i} \leftarrow c_{i+1} \wedge x_{i}\\
c_{i+1} \leftarrow 0\\
s_{i}^0 \leftarrow \itxss{c_{i+1} \wedge \neg x_i }{
	s_{i+1}^0 \vee s_{i+1}^1
}{s_{i}^0}\\
s_{i}^1 \leftarrow \itxss{c_{i+1} \wedge x_i }{
	s_{i+1}^0 \vee s_{i+1}^1
}{s_{i}^1}
\end{cases}$
\end{tabular}

\noindent\textbf{Step $3+n$.}\\ Set every variable to $1$ if QBF satisfied. After this step, the program returns to $f_1$.\\
$f_{3+n}(\cdot) =\begin{cases}v \leftarrow \itx{s_{1}^0 \wedge s_{1}^1}{1}{v} &(\text{for all variables } v)
\end{cases}$ \\
The (3+n)th step  ensures that configuration $(1,\ldots,1)$ will be reached if and only if the given QBF formula is satisfied.
\end{proof}

\begin{remark}[{Choice of rounding function}]\label{remark:choiceoffn}
 The presentation here relies on specific choices of rounding function, but we observe that the language can easily exchange several different natural rounding functions, so the reduction is robust.  The rounding is only useful in the \texttt{and} and \texttt{or} instructions. The floor function can be replaced by essentially any other rounding. For example $x_j\vee x_\ck = \ceil{\frac{x_j + x_\ck}{2} }$ and $x_j\wedge x_\ck = \ceil{\frac{-1 + x_j + x_\ck}{2} }$. Similarly, when $[\cdot]$ is minimal error rounding then  $x_j\vee x_\ck = \round{\frac{1 + x_j + x_\ck}{3} }$ and $x_j\wedge x_\ck = \round{\frac{x_j + x_\ck}{3} }$ (the break point is not used). Thus, the problem will also be hard for any of these roundings.
 \end{remark}
 \begin{remark}[{Perturbation: ensuring the eigenvalues are not modulus 1}]\label{remark:purtubation}
Observe that under the perturbation that multiplies each operation by 1.1 (before taking floor) we obtain the same resulting operation. For example $x_i \leftarrow x_j\vee x_\ck = \floor{\frac{1 + x_j + x_\ck}{2} }$ is equivalent to $x_i \leftarrow x_j\vee x_\ck = \floor{(\frac{1 + x_j + x_\ck}{2} )*1.1}$. Hence, if the resulting matrix $M$ has eigenvalues $1$, taking $1.1 M$ (or similar value to 1.1) will result in a matrix that does not with the same orbit; which shows that hardness is retained for matrices in which no eigenvalue has modulus 1.\end{remark}
\begin{remark}[Dimension]The hardness result needs reachability instances of unbounded dimension. For a QBF formula with $n$ variables and $\ell$ logical operations, the resulting instance of rounded P2P has dimension $(3n+1+\ell)(4n+15+\ell)$.\end{remark}

\section{Special cases on non-hyperbolic systems}

In this section we consider certain cases when the eigenvalues can be of modulus one. In particular we work in the Jordan normal form and show that the problem can be solved for certain types of rounding. We fall short of arbitrary deterministic rounding, which would be required to show the problem in full generality through the Jordan normal form approach.

First, we show  decidability for polar-rounding, along with an example with numbers requiring exponential space by the time the system becomes periodic---seeming to imply any `wait and see' approach would require $\EXPSPACE$. We also show decidability for certain types of Argand rounding, in particular truncation and expansion, but minimal-error rounding remains open (which we discuss further in \cref{sec:problemswith2x2}).

\subsection{Polar rounding with updates in Jordan normal form}\label{sec:polar}

We restrict ourselves to a Jordan block $M$ of dimension $d$, with eigenvalue $\lambda$ of modulus $1$. Since the polar rounding function has bounded effect \textit{on the modulus}, the remaining blocks, which need not be of modulus $1$ can be solved (\cref{thm:jnf-bounded-rounding}) by running this algorithm in lock step with the algorithm for those blocks.
All together, this gives us:

\begin{theorem}\label{thm:polar-expspace}
The Rounded P2P Problem is decidable in $\EXPSPACE$ for the polar rounding function with $\theta_g = \frac{\pi}{\numangles}$, $\numangles{} \ge 2$ and matrices $M \in \algs^{\dimension \times \dimension}$ in Jordan normal form. 
\end{theorem}

To prove~\Cref{thm:polar-expspace} we show that each dimension $d, d{-}1, \ldots, 1$ will eventually be periodic on a fixed modulus, or permanently diverge beyond $\target_\ck$ (the target value in dimension $\ck$).

Let $\abr{a,b}$ be the smallest angle between vectors $a$ and $b$ -- this is a value in $[0,\pi]$ and, in particular, it is always positive. It is used as a measure of alignment: the more $a$ and $b$ are aligned the smaller $\abr{a,b}$ is. We will assume that the system will round up if $[x] - x = 0.5$.
The remaining case can be adapted by suitably adjusting the relevant inequalities. 
 We say that a dimension $\ck \in \{1,\ldots,d\}$ is \emph{just rotating} after position $N$, if for all $i \geq N$: $\vecitcomp{x}{i+1}{\ck} = [\lambda \, \vecitcomp{x}{i}{\ck}]$.
Note that dimension $d$ is just rotating after $0$, by definition.
Our goal is to show that every dimension $\ck$ will eventually be just rotating (for which we would require it to have modulus $\abs{\target_\ck}$) or reach a point that lets us conclude it has permanently diverged past $\target_\ck$.
So we assume, henceforth, that dimension $\ck$ is just rotating.

We let $\phi(i) = \abr{\lambda \vecitcomp{x}{i}{\ck{-}1}, \vecitcomp{x}{i}{\ck}}$.
As $\vecitcomp{x}{i+1}{\ck{-}1} = [\lambda \vecitcomp{x}{i}{\ck{-}1} + \vecitcomp{x}{i}{\ck}]$, small values of $\phi(i)$ (between $0$ and $\pi/2$) lead to an increase in modulus of $\vecitcomp{x}{i+1}{\ck{-}1}$, whereas large values (between $\pi/2$ and $\pi$) lead to a decrease when $\abs{\vecitcomp{x}{i}{\ck{-}1}}$ is sufficiently large relative to $\abs{\vecitcomp{x}{i}{\ck}}$.
Our analysis relies on the fact that $\phi(i)$ can never increase:
\begin{restatable}{lemma}{angledecreasing}
  \label{lemma:angledecreasing}
  \label{corollary:phidecreasing}
  Suppose that dimension $k$ is just rotating after step $N$. Then, for all $i \geq N +1$: $\phi(i) \geq \phi(i+1)$.
\end{restatable}
If dimension $\ck{-}1$ repeats its relative angle to $\ck$ and its modulus in some step, we can conclude that $\ck{-}1$ is just rotating:
\begin{restatable}{lemma}{willperiod}\label{lemma:willperiod}
  Suppose that dimension $\ck$ is just rotating after step $N$, that
$\phi(N) = \phi(N{+}1)$ and $\abs{\vecitcomp{x}{N}{\ck{-}1}} = \abs{ \vecitcomp{x}{N+1}{\ck{-}1}}$. Then, dimension $k{-}1$ is just rotating after $N$.
\end{restatable}
If the precondition of~\Cref{lemma:willperiod} holds, we move to the next dimension $\ck{-}2$.
Otherwise, we want to give a bound such that whenever $\abs{\vecitcomp{x}{i}{\ck{-}1}}$ exceeds it, we can conclude that it never decreases back to $\abs{\target_{\ck{-}1}}$.
We first introduce the angle $\gamma(i) = \abr{\lambda \vecitcomp{x}{i}{\ck{-}1} + \vecitcomp{x}{i}{\ck},\lambda \vecitcomp{x}{i}{\ck{-}1}}$.
The angle $\gamma(i)$ decreases with increasing $\abs{\vecitcomp{x}{i}{\ck{-}1}}$, as dimension $\ck$ is just rotating and hence does not change in modulus.
We observe that $\gamma(i) \leq \phi(i)$ for all $i$.
The following shows that an increase in modulus caused by crossing an `axis' (i.e. if $\gamma(i) > \pi/2$) can only happen once, as in the next step, the angle will have decreased.
\begin{restatable}{lemma}{ulq}\label{lemma:ulq}
Let $a = \lambda\vecitcomp{x}{i}{\ck{-}1}$ and $b = \vecitcomp{x}{i}{\ck}$.
Suppose that $\theta_g \le \pi/2$, $\gamma(i) > \pi/2$ and $\abs{a+b} > \abs{a}$.
Then $\abr{\lambda[ a+ b ], [\lambda b]} \le \pi/2$, entailing  $\gamma(i+1) \leq \phi(i+1) \le \pi/2$.
\end{restatable}
Furthermore, a decrease cannot be followed by an increase, unless the angle changes:
\begin{restatable}{lemma}{nodecreaseincrease}\label{corr:nodecreaseincrease}
Suppose dimension $\ck$ is just rotating after $N$. It is not possible for $i-1\ge N$, to have $\phi(i-1)=\phi(i) = \phi(i+1) > \pi/2 $ and $\abs{\vecitcomp{x}{i+1}{\ck{-}1}} > \abs{\vecitcomp{x}{i}{\ck{-}1}} < \abs{\vecitcomp{x}{i-1}{\ck{-}1}}$.
\end{restatable}
Finally, we place a limit on the number of consecutive increases until we can decide that dimension $\ck{-}1$ will not decrease below the current modulus in the future:
\begin{restatable}{lemma}{notdecreasingagain}
  \label{cor:notdecreasingagain}
    Let $a = \lambda \vecitcomp{x}{N}{\ck{-}1}$ and $b = \vecitcomp{x}{N}{\ck}$ for some $N > 0$.
    Suppose that $\ck$ is just rotating after $N$, $\abs{a+b} > \abs{a} + 0.5$ and $\abs{a} \geq \frac{1}{\sqrt{2}}\abs{b}$. Then, for all $i > N$: $\abs{\vecitcomp{x}{i}{\ck{-}1}} > \abs{a}$.
\end{restatable}
With Lemmata \ref{lemma:willperiod}-\ref{cor:notdecreasingagain} we are in a position to prove~\Cref{thm:polar-expspace} (the proofs of the preceding lemmata, and the $\EXPSPACE$ analysis can be found in \cref{appen:sec:polar}).

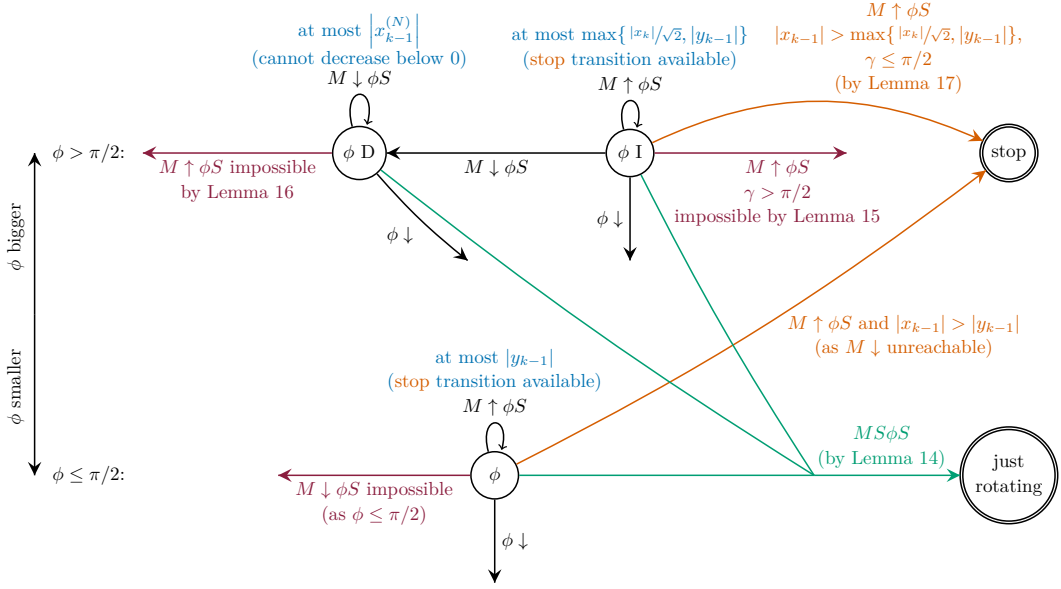
\begin{figure}[t]
\centering

\resizebox{\textwidth}{!}{%

\begin{tikzpicture}[auto,node distance=2.8cm,
                    thick]
  \tikzstyle{every state}=[]
  \tikzset{myptr/.style={decoration={markings,mark=at position 1 with %
    {\arrow[scale=1.4,>=stealth,very thick]{>}}},postaction={decorate}}}

  \node[state] (A) at (4,0) {$\phi$ I};
  \node[state] (A1) at (-1,0) {$\phi$ D};

  \node[state,accepting] (stop) at (11,0) {stop};

  \node (periodic) at (9,-6) {};

  \node[state,accepting,align=center] (jr) at (11,-6) {just\\ rotating};

  \node[state] (B) at (1.5,-6) {$\phi$};

  \node[right] at (-6.8,0) {$\phi > \pi/2$:};
  \node[right] at (-6.8,-6) {$\phi \le \pi/2$:};

  \path 

  (-7,-3) edge[myptr] node[rotate=90,yshift=10,xshift=20] {$\phi$ bigger} (-7,0)

  (-7,-3) edge[myptr] node[rotate=90,yshift=10,xshift=-20] {$\phi$ smaller} (-7,-6)

  (A) edge[myptr,swap] node {$\phi \downarrow$} ($(A) + (0,-2)$)
  (A1) edge[myptr,swap,bend right=5] node {$\phi \downarrow$} ($(A1) + (2,-2)$)

  (A1) edge[loop above] node(Z) {$M\downarrow\phi S $} (A1)
  (A) edge[loop above] node (Z1) {$M\uparrow\phi S $} (A)

  (A) edge[myptr] node {$M\downarrow\phi S $} (A1)
  (A) edge[myptr, bend left=25,ApyVermillion] node[align=center,xshift=1.5cm] {$M\uparrow\phi S $\\ $\abs{x_{k-1}} > \max\{\sfrac{\abs{x_{k}}}{\sqrt{2}}, \abs{\target_{k-1}}\}$,\\ $\gamma \le \pi/2$\\ (by \cref{cor:notdecreasingagain})} (stop)

  (B) edge[loop above] node (Z2) {$M\uparrow\phi S $} (B)
  (B) edge[myptr] node {$\phi \downarrow$} ($(B) + (0,-2)$)

    (A) edge[myptr,MRed,swap] node[align=center,xshift =5mm] {$M\uparrow\phi S$\\ $\gamma >\pi/2$\\ impossible by \cref{lemma:ulq}} ($(A) + (+4,0)$) 

  (A1) edge[myptr,MRed] node[align=center] {$M\uparrow\phi S$ impossible\\ by \cref{corr:nodecreaseincrease}} ($(A1) + (-4,0)$) 

  (B) edge[myptr,MRed] node[align=center] {$M\downarrow\phi S$ impossible \\ (as $\phi \le \pi/2$)} ($(B) + (-4,0)$) 

  (B) edge[myptr,bend right=5,swap,ApyVermillion] node[align=center,yshift=4mm,xshift=4mm] {$M\uparrow\phi S$ and $\abs{x_{k-1}}  > \abs{\target_{k-1}}$\\ (as $M\downarrow$ unreachable)} (stop)

  (A) edge[ApyBluishGreen, bend right=3] node {} (7.4,-6)
  (A1) edge[ApyBluishGreen,,swap, bend right=3] node {} (7.4,-6)
  (B) edge[ApyBluishGreen,] node {} (7.4,-6)
  (7.4,-6) edge[ApyBluishGreen,myptr] node[align=center,xshift=-1mm] {$MS\phi S$\\ (by \cref{lemma:willperiod})} (jr)

  ;

  \node[ApyBlue,align=center] at ($(Z) + (0,0.7)$) {at most $\abs{x^{(N)}_{k-1}}$\\(cannot decrease below $0$)};
  \node[ApyBlue,align=center] at ($(Z1) + (0,0.7)$) {at most $\max\{\sfrac{\abs{x_k}}{\sqrt{2}},\abs{\target_{k-1}}\}$\\({\color{ApyVermillion}stop}  transition available) };
  \node[ApyBlue,align=center] at ($(Z2) + (0,0.7)$) {at most $\abs{\target_{k-1}}$ \\({\color{ApyVermillion}stop} transition available)};

\end{tikzpicture}

}%

\caption{State diagram for $\phi$ whilst considering dimension $\ck{-}1$, assuming $\ck$ is \textit{just rotating}.}
\label{fig:state-machine}
\end{figure}

\begin{proof}[Proof of \cref{thm:polar-expspace}]
As described above, we consider each dimension separately, starting with $\ck{} = \dimension$, and assume by induction that the previous dimension is just rotating.
We describe an algorithm that tracks the value of $\phi$ and operates according to \cref{fig:state-machine}.
Each realisable value of $\phi$ relates to a copy of \cref{fig:state-machine} (we only draw one example of $\phi$ satisfying $\phi > \pi/2$ and $\phi \leq \pi/2$ respectively).
For $\phi > \pi/2$ two states are used, one which encodes that the previous transition was decrementing the modulus ($\phi$ D), the other which indicates the previous was not decrementing (including first arrival) ($\phi$ I).

The algorithm moves on each update step according to the arrow, which denotes whether the update is modulus increasing $M\uparrow$, decreasing $M\downarrow$ or stationary $MS$.
Similarly $\phi$ may decrease $\phi\downarrow$ or stay stationary $\phi S$, but never increase (\cref{corollary:phidecreasing}).
Whenever $\phi$ decreases we make progress through the DAG to a lower value of $\phi$.
All combinations $\{M\uparrow,M\downarrow,MS\}\times \{\phi\downarrow,\phi S\}$ are accounted for at each state.

Progress is made whenever we move through the DAG towards a stopping criterion.
For self-loops a bound is provided (in blue) on the maximum time spent in this state.
Since for each dimension we will ultimately end up in just rotating, or be able to stop early, the problem is decidable. 
\end{proof}

\begin{figure}[t]
\begin{subfigure}{.5\textwidth}
\centering
\includegraphics[width=0.95\textwidth]{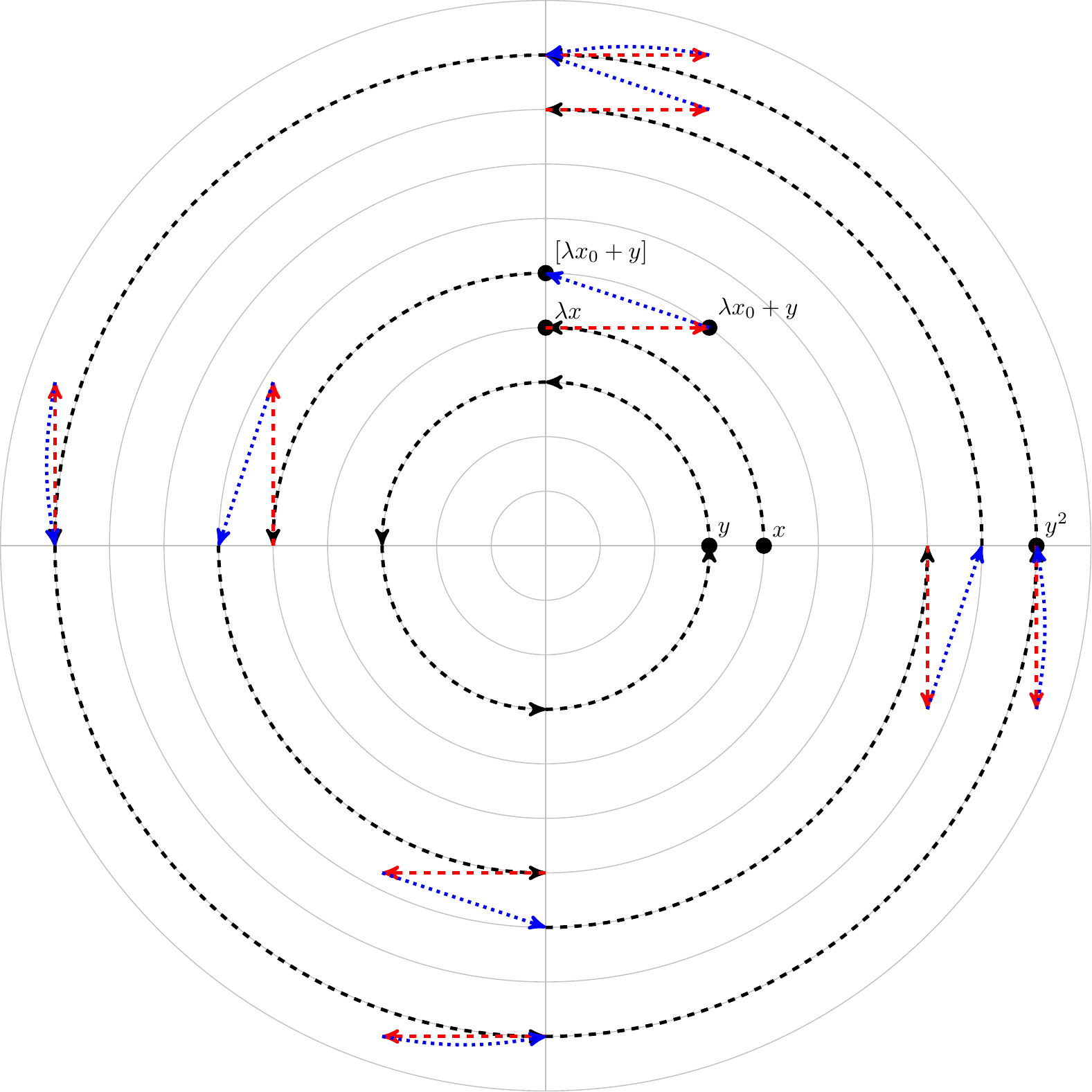}

\end{subfigure}
\begin{subfigure}{.5\textwidth}
\includegraphics[width=0.95\textwidth]{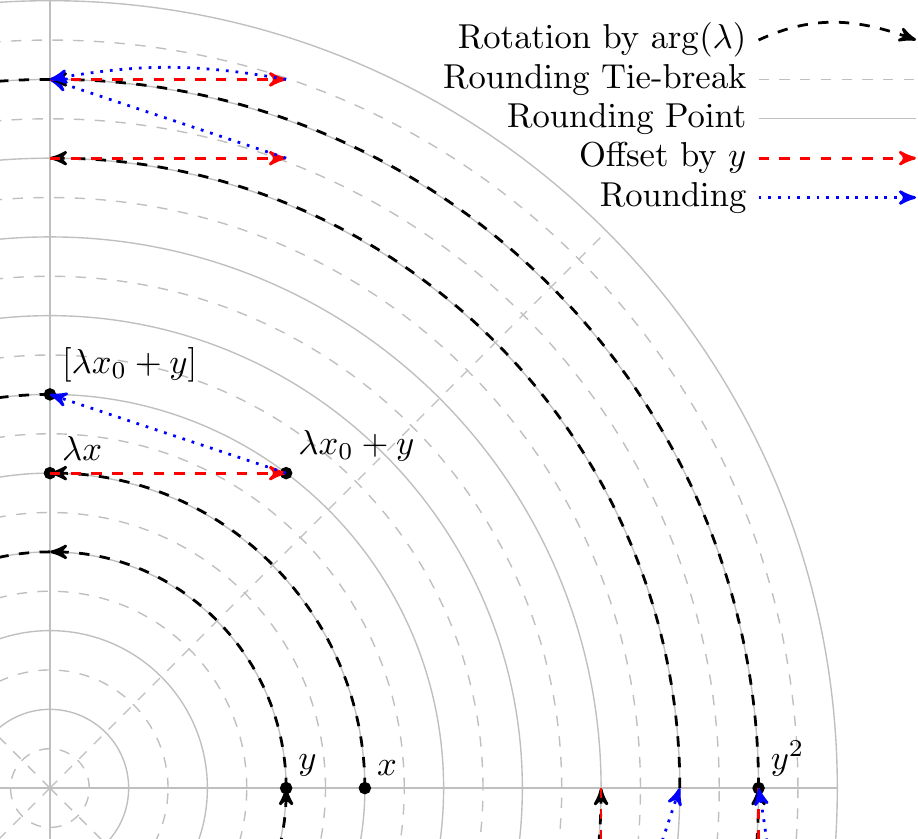}
\end{subfigure}
\caption{Example where the system may become large before being periodic (see~\cref{eg:expspace-req}). }
\label{fig:eg-large-period}
\end{figure}

\begin{example}[System requiring $\EXPSPACE$ to be periodic]\label{eg:expspace-req}
If $\phi\le \pi/2$, the considered dimension will either diverge at some point, or become periodic.
This depends, essentially, on whether $\abs{\vecitcomp{x}{i}{\ck}}\cos(\phi)< 0.5$, in which case the rounding will not lead to an increase when $\abs{\vecitcomp{x}{i}{\ck{-}1}}$ is sufficiently large relative to $\abs{\vecitcomp{x}{i}{\ck}}$.
We give an example where $\abs{\vecitcomp{x}{i}{\ck{-}1}}$ grows to $\abs{\vecitcomp{x}{i}{\ck}}^2$, and requires numbers of doubly exponential size (and exponential space) in $\dimension$ before becoming periodic.
We assume that $\theta_g = \pi/2$ (so there are four possible angles) and integer modulus granularity.
Let $M$ be a single Jordan block of dimension $\dimension$ with eigenvalue $\lambda = e^{i\pi/2}$. The angle $\phi(i)$ remains constant, but the modulus grows while $\abs{\vecitcomp{x}{i}{\ck}} < \abs{\vecitcomp{x}{i}{\ck{-}1}} \le \abs{\vecitcomp{x}{i}{\ck}}^2$.
We start at the point $\vecit{x}{0} = ((3 + \dimension,0), \dots,(6,0),(5,0),(4,0))$, using the representation that $Ae^{i\theta}$ is written $(A,\theta)$. This system is periodic, with maximal component $\vecit{x}{N} = ((4^{(2^{\dimension-1})},0), \dots,(4^{2\cdot 2\cdot 2},0),(4^{2\cdot 2},0),(4^2,0),(4,0))$. Note that $4^{2^4}$ is larger than a 32-bit number.
This idea is illustrated in~\cref{fig:eg-large-period}, where $y$ represents $\vecitcomp{x}{i}{\ck}$ and is just rotating, and $x$ represents  $\vecitcomp{x}{i}{\ck{-}1}$, which grows to $\abs{y}^2$.
\end{example}

Despite \cref{eg:expspace-req}, which shows that waiting until becoming periodic may need exponential space, we conjecture the Rounded P2P can be solved in $\PSPACE$.
This is because if $\vecitcomp{x}{i}{1}$ exceeds a value representable in polynomial space we expect it will never return to the target $\target_1$ (a value representable in polynomial space).
However, we are unable to show at the moment that it never gets very large and subsequently returns to a small value.

\subsection{Argand truncation or expansion in Jordan normal form}\label{sec:truncated}
We now consider Argand truncation based rounding showing decidability in $\EXPSPACE$. The rounding function is of the form $[a+bi] = [a]+[b]i$ where, for $x\in \mathbb{R}$ ,  $[x] =\floor{x}$ if $x \ge 0$ and $[x] = \ceil{x}$ if $x < 0$, which has a non-increasing effect on the modulus.

\begin{restatable}{theorem}{jnftruncation}\label{thm:jnf-tw-zero}
The Rounded P2P Problem is decidable in $\EXPSPACE$ for deterministic Argand rounding function with a non-increasing effect on the modulus and matrices $M \in \algs^{\dimension \times \dimension}$ in Jordan normal form.
\end{restatable}
As a key ingredient of \cref{thm:jnf-tw-zero} we will make use of the following theorem:

\begin{theorem}[{\cite[Corollary 3.12, p.41]{niven1956irrational}}]\label{thm:nivens}
Both $x$ is a rational multiple of $\pi$ and $\sin(x)$ is rational only at $\sin(x) = 0, \frac{1}{2}, \text{ or } 1$. Both $x$ is a rational multiple of $\pi$  and $\cos(x)$ are rational only at $\cos(x) = 0, \frac{1}{2},$ or $1$. Both $x$ is a rational multiple of $\pi$  and $\tan(x)$ are rational only at $\tan(x) = 0,$ or $\pm1$.
\end{theorem}

\begin{proof}[Proof sketch of \cref{thm:jnf-tw-zero}]
Without loss of generality we consider only a single Jordan block with $\abs{\lambda} = 1$, as the remaining blocks can be handled in lock step (using the algorithm of \cref{thm:jnf-bounded-rounding} if the eigenvalue is not of modulus one).
Consider the $\dimension$th component. At each step, whenever rounding takes place, then there is some decrease in the modulus. Thus, either the coordinate hits zero (and stays forever), or it stabilises and becomes periodic (with no rounding ever occurring again). 
The $\dimension$th coordinate can be simulated until this happens.
At this point, if its modulus is not $\abs{\target_\dimension}$, $\target$ will not be reached in the future and we return $\textsc{no}$. 

If dimension $x_\dimension$ reaches zero, then this dimension from some point on becomes irrelevant and the instance can be reduced to an instance of dimension $d-1$. Note that this case must occur if $\arg(\lambda)$ is not a root of unity as an irrational point is found infinitely often.

In the case where $x_\dimension$ does not reach zero, then it is periodic at some modulus. This implies it never rounds again, and so surely hits integer points at every step. We show that this can only occur if $\arg(\lambda)$ is a multiple of $\pi/2$.  Assume that $\arg(\lambda)$ is not a multiple of $\pi/2$: the rotation of a point with integer coordinate to integer coordinate leads to the conclusion of either rational tangent or rational sine and cosine. By \cref{thm:nivens} a rational tangent alongside a rational angle ($\arg(\lambda)$ is a root of unity) implies that the angle must be a multiple of $\pi/4$. It is not $\pi/4$, as there is no Pythagorean triangle with angle $\pi/4$. By \cref{thm:nivens} rational sine and cosine and rational angle concludes the angle must be a multiple of $\pi/2$.
Finally, we show that when $\arg(\lambda)$ is a multiple of $\pi/2$ the system surely diverges at dimension $\dimension{-}1$, and hence we can put a bound on how far we need to simulate. 
\end{proof}

\begin{remark*}[{Argand expansion in Jordan Normal Form}]
Instead of considering the rounding function to always decrease the modulus, we consider the rounding function to always increase the modulus. Then, by the same rationality argument either $\arg(\lambda)$ is a multiple of $\pi/2$ (so no rounding occurs and standard methods can be applied), or $\arg(\lambda)$ is not a multiple of $\pi/2$ and rounding is applied infinitely often. We observe that rounding infinitely often results in divergence. Suppose instead the modulus converges, in supremum, to $C$. However the $[C]$-ball is finite, thus rounding infinitely often must eventually exhaust the set, contradicting supremacy.  Since divergence occurs in the $d$th component the system can be iterated until either $\vecit{x}{i} = y$ or $\vecitcomp{x}{i}{\dimension}$ exceeds $y_\dimension$. (Unless $\vecitcomp{x}{0}{\dimension} = \target_\dimension=0 $, in which case the $d$th component can be deleted.)
\end{remark*}
\section{Discussion of open problems}\label{sec:problemswith2x2}
\begin{sloppypar}
In this section we consider the following open problem, which already exhibits a technical difficulty for a relatively simple instance.

\begin{openproblem}\label{openproblem:klabr}
Under which deterministic bounded-effect rounding functions does the Rounded P2P Problem become decidable (even when restricted to Jordan normal form)?
\end{openproblem}

In particular we emphasize that even decidability of the Rounded P2P Problem in the case of a 2D \textit{rotation} matrix remains open. This should be compared to the papers \cite{LowensteinHV1997, BosioV1999, VivaldiV2003, Vivaldi2006, AkiyamaP2013, PethoTW2019}, which consider linear maps on $\reals^2$ that are close to rotations, and the floor rounding $\lfloor \cdot\rfloor$ is used to induce discretized maps on $\ints^2$. The conjecture made in \cite{Vivaldi2006} that all orbits of these maps are eventually periodic (and thus finite) is, to the best of our knowledge, still open in general. This lack of understanding of the dynamics of rotations even on a $2$-dimensional lattice is striking and hints at an intrinsic level of difficulty in dealing with eigenvalues of modulus $1$.

\end{sloppypar}

\begin{figure}\centering
\noindent\begin{tabular}{@{}c@{}c@{}c@{}c@{}}
	\includegraphics[width=0.243\textwidth]{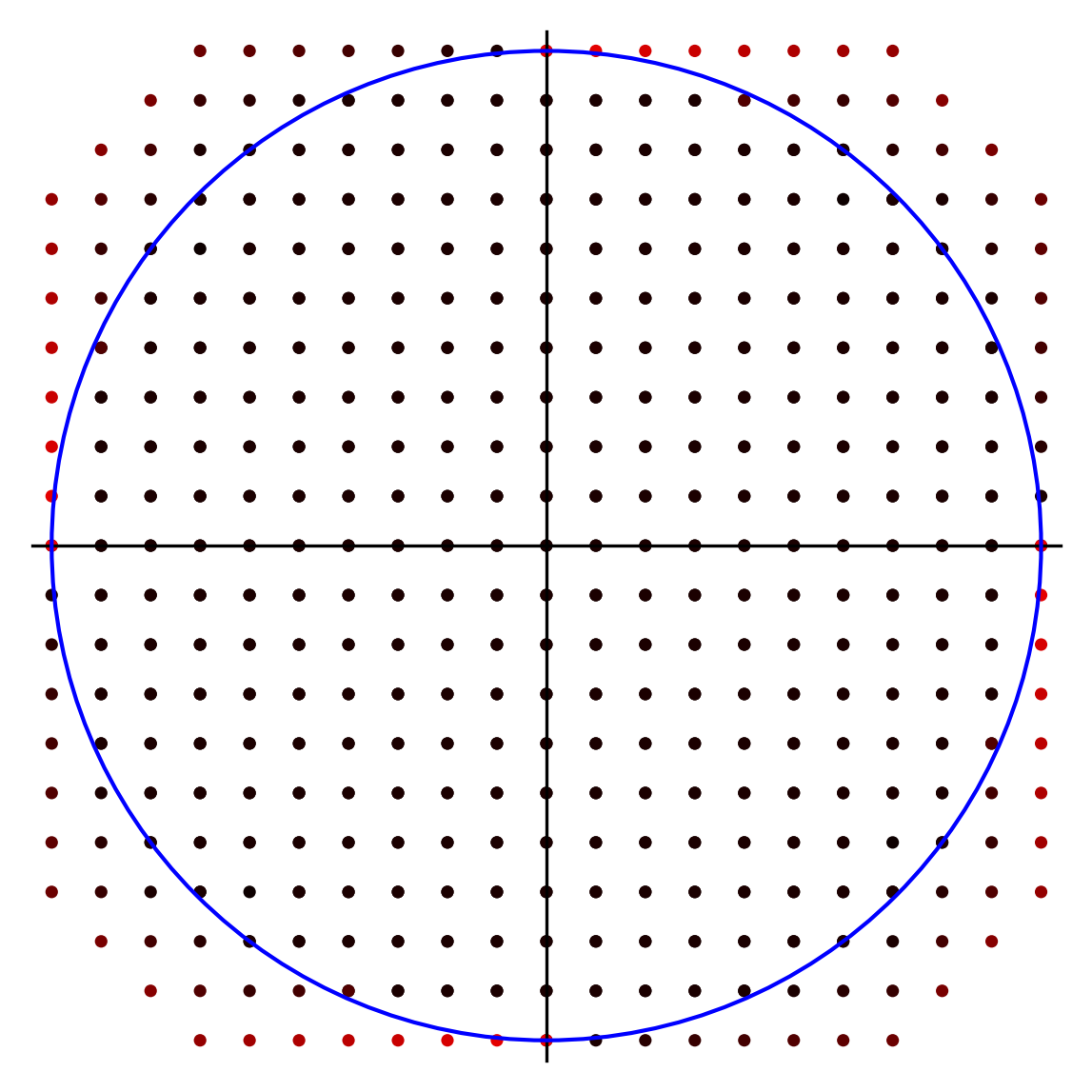}&%
	\includegraphics[width=0.243\textwidth]{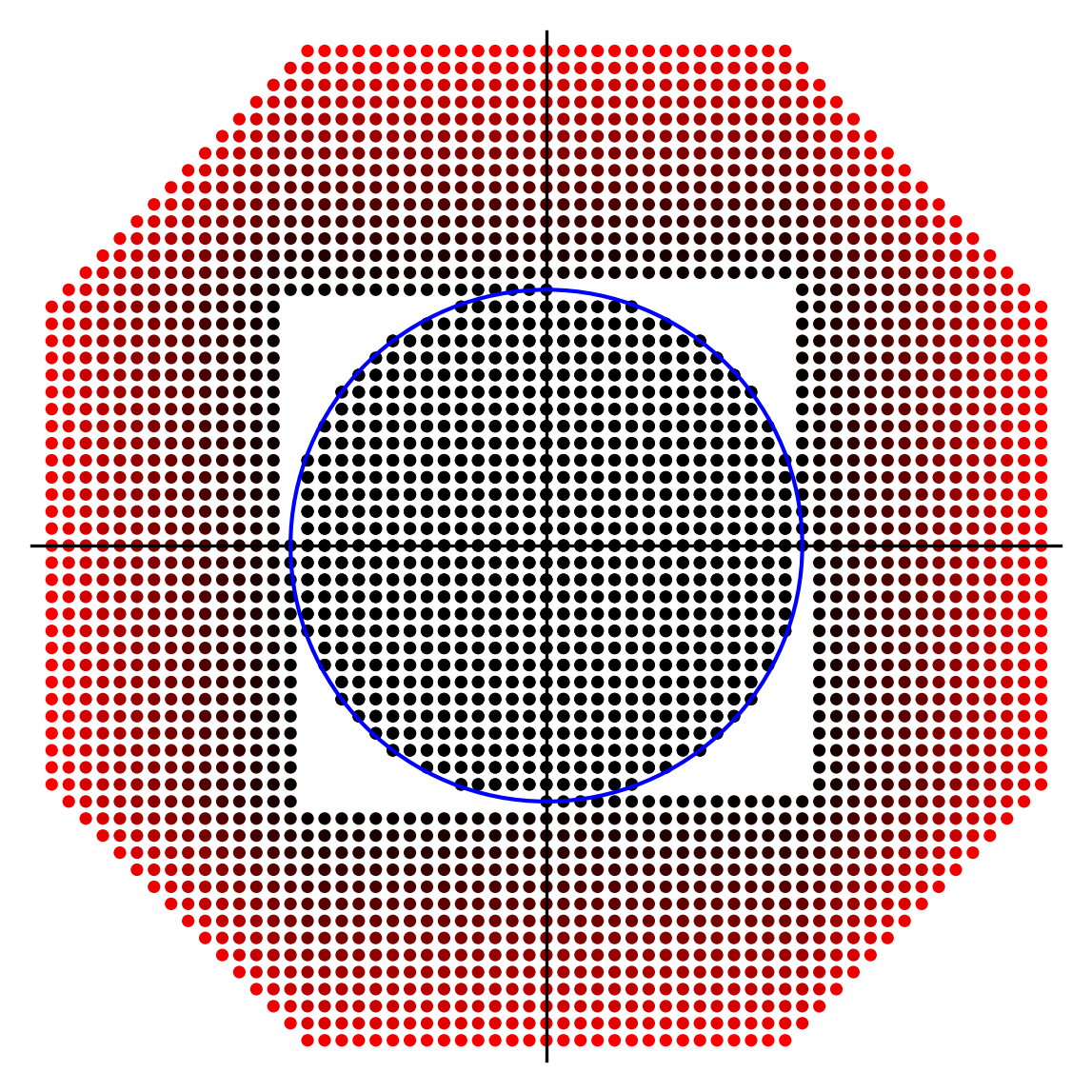}&%
	\includegraphics[width=0.243\textwidth]{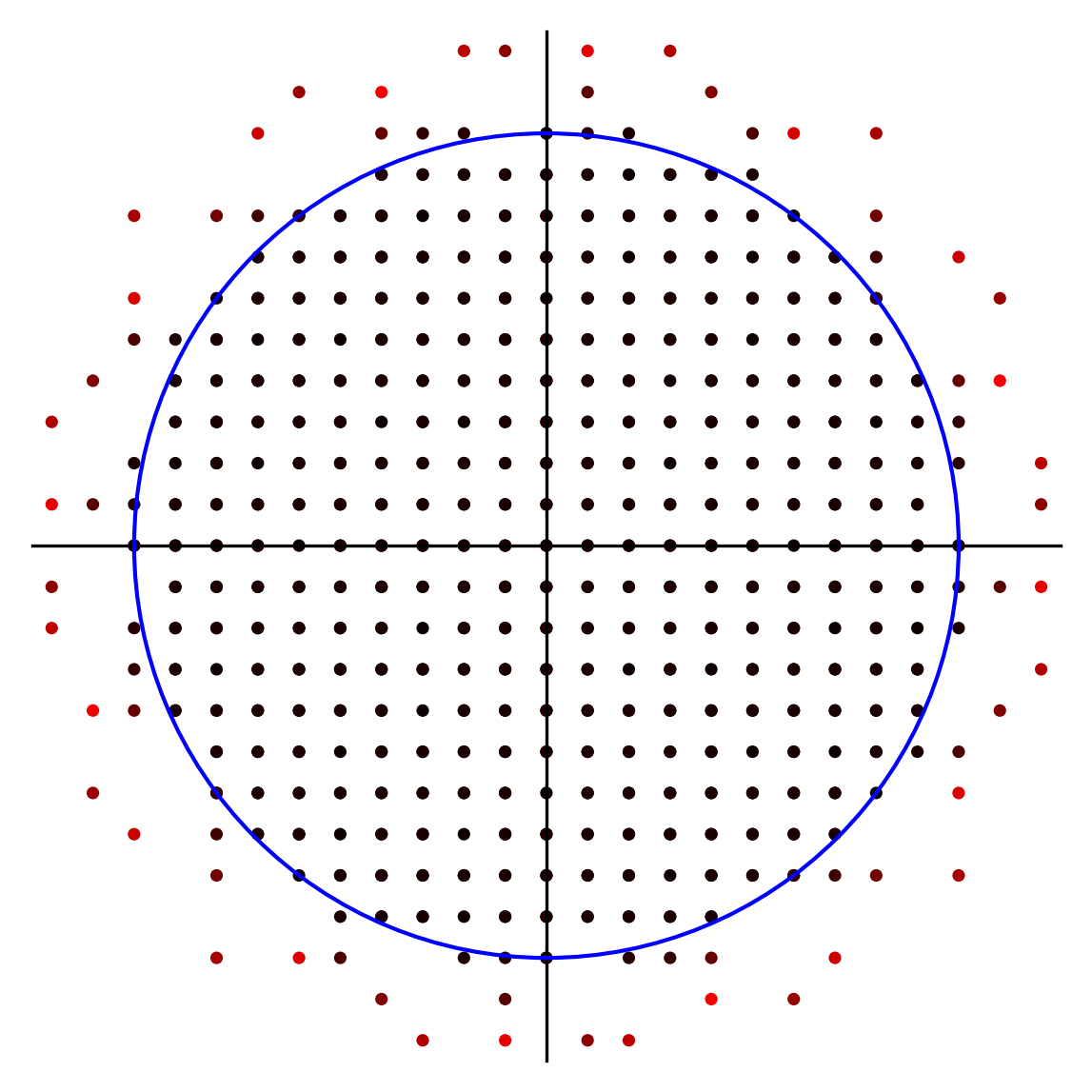}&%
	\includegraphics[width=0.243\textwidth]{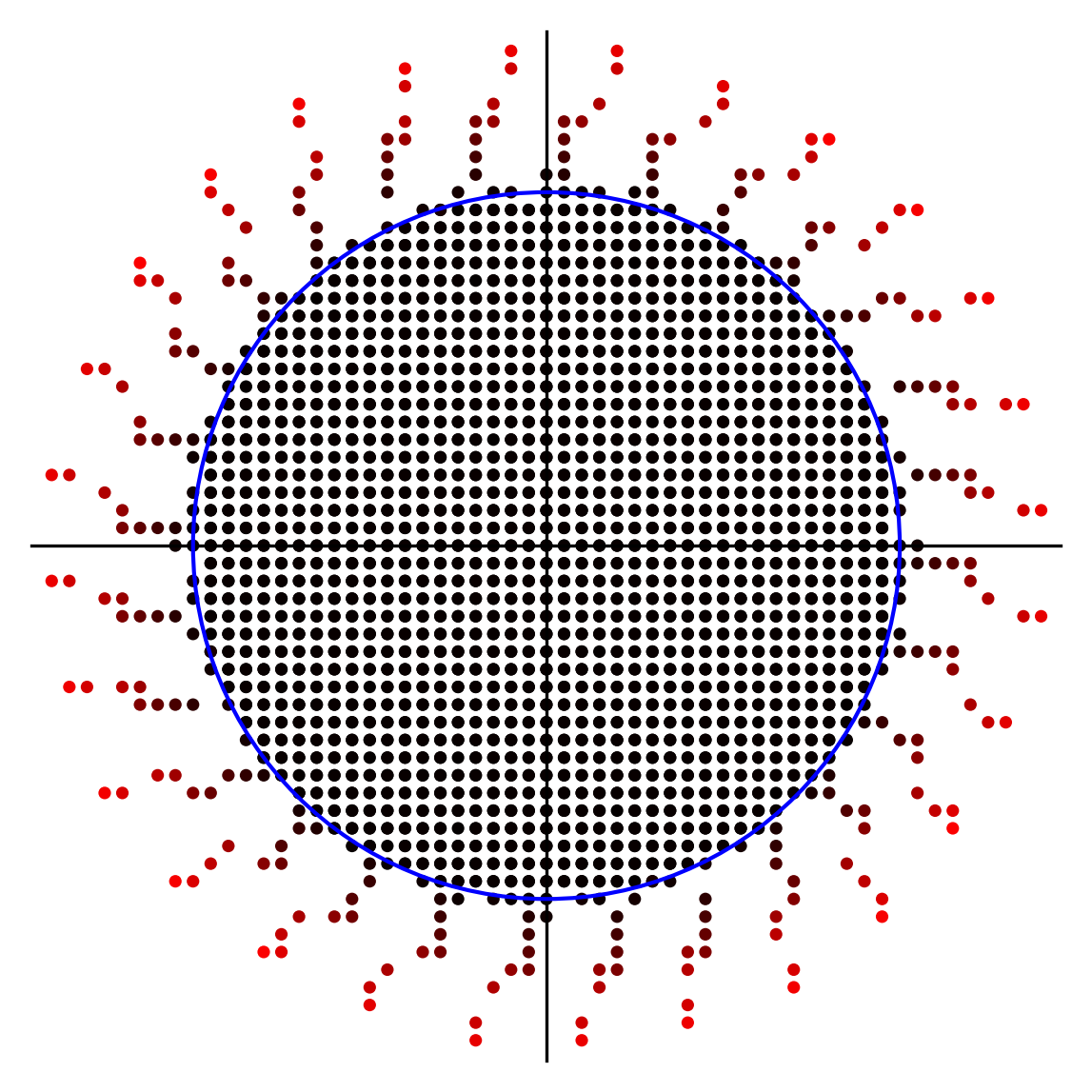}\\
		(a) $r=10$, $\theta=\pi/42$
	&
	(b) $r=15$, $\theta=\pi/91$
	&
	(c) $r=10$, $\theta=\frac{2^{(0.4)}}{10}\pi$
	&
	(d) $r=20$, $\theta=\pi/14$
\end{tabular}
\caption{Rotational examples. We start with all points in the circle of radius $r$, and consider the effect of rotating every point by $\theta$, followed by minimal error rounding. This can be seen as viewing the combined orbits, starting at several points. Redder points are added in later generations.}\label{fig:rotations}
\end{figure}

We ran experiments on the behaviour of rounded orbits induced by rotations in the plane. Four prototypical results are depicted in \cref{fig:rotations}. We note that in every one of our examples  the orbits  eventually become  periodic. Moreover, all experiments fall into the four categories of \cref{fig:rotations}, i.e., where the resulting set consists of (a) a square with cut-off corners, (b) this same square, but with a central square cut out, and  (c) all points within the circle with some seemingly randomly added points outside (in the case of an irrational multiple of $\pi$), (d) the initial circle with added `tentacles' occuring in intervals corresponding to the rotational angle (in the case of a rational multiple of $\pi$). We have been unable construct a rotation with an infinite rounded orbit.

One could hope that other kinds of rounding functions simplify the analysis of the orbits. We have shown truncation based rounding, for example, either helps converge towards zero, or diverge towards infinity, and this can be exploited (particularly at the bottom dimension of a Jordan block). However, roundings which may either round up or down greatly complicate the analysis. 
Nevertheless, we conjecture that all rounded orbits obtained by rotation eventually become periodic.

\subparagraph*{Random rounding functions}
Orbit problems for rounding functions which behave probabilistically are another line of open problems and are a natural candidate for future work.

\appendix

\section{Additional material for Section~\ref{sec:pspacehard}, \texorpdfstring{$\PSPACE$}{PSPACE}-hardness}

\subsection{Perturbation}
We expand on \cref{remark:purtubation}, observing that multiplying by $1.1$ (or similar) maintains the logical equivalence required for all of our update functions. Assuming $x_i,x_j,x_k$ are each in $\{0,1\}$ we have,
\begin{itemize}
	\item $1 \leftarrow \floor{1 * 1.1}$
\item $x_i \leftarrow x_j\vee x_\ck = \floor{\frac{1 + x_j + x_\ck}{2} } =\floor{(\frac{1 + x_j + x_\ck}{2} )*1.1}$
\item $x_i \leftarrow x_j\wedge x_\ck =\floor{\frac{1 + x_j + x_\ck}{3} }  \floor{(\frac{1 + x_j + x_\ck}{3} )*1.1}$ 
\item $x_i \leftarrow \neg x_j = \floor{1-x_j} = \floor{(1-x_j)*1.1}$.
\item $x_i \leftarrow \floor{0} = \floor{(0)*1.1}$
\item $x_i \leftarrow \floor{x_i} = \floor{(x_i)*1.1}$
\item $x_i \leftarrow \floor{x_j} = \floor{(x_j) *1.1}$.
\end{itemize}

\subsection{Dimension of Rounded P2P instance in proof of \texorpdfstring{$\PSPACE$}{PSPACE}-hardness}
It is clear that the required reduction is polynomial, we precisely characterise the dimension of the resulting system here.
\begin{proposition}\label{prop:rkldim}
The resulting instance of rounded P2P has dimension $(3n+1+\ell)(4n+15+\ell)$, if $\psi$ has $\ell$ logical operations.
\end{proposition}
\begin{proof}[Proof of \cref{prop:rkldim}]The functions $f_1, \dots, f_{3+n}$ hide the inner workings of the reduction to the Rounded P2P Problem, by contracting steps and auxiliary variables and illustrating the effect using logic, rather than the floor of a linear combination.

Following the routine steps to translate the logical commands into the Rounded P2P Problem, we see that: 
\begin{itemize}
\item $f_1$ depends on the formula $\psi$ to evaluate. If $\psi$ is a formula with $\ell$ logical operators we have:
\begin{itemize}
  \item $\ell$ steps, resolving each logical operator according to topological ordering
  \item $\ell$ auxiliary variables to store partial computations.
  \end{itemize}
\item $f_2$ takes two steps and four extra variables.

\item $f_{3+n-i}$ takes $3$ steps  for each $i$ and $8$ extra variables. The 8 variables can be shared for all functions.

\item $f_{3+n}$ takes 3 steps and $2\times t$ extra variables, where $t$ is the total number of main variables. However it can be simplified to 2 steps, and 1 extra variable (by noticing it is equivalent to $v \leftarrow (s_{1}^0 \wedge s_{1}^1) \vee v$).
\end{itemize}
Thus the total number of steps is $\ell+2+3(n-1) +2 = 3n+1+\ell$ steps. The total number of variables is $\ell + t + 4 + 8+1$ plus $1$ to store \texttt{true}, so total of $t+14+\ell$. Note that $t = 4n+1$, total $4n +15+\ell$. Thus when exploded as per \cref{sec:explode-dim}, there are $(3n+1+\ell)(4n+15+\ell)$ dimensions in the Rounded P2P Problem.
\end{proof}

\section{Additional material for Section~\ref{sec:polar}, Polar rounding in Jordan normal form}
\label{appen:sec:polar}

\begin{lemma}\label{lemma:rotationorderpreserving}
Assume $b = [b]$, then $\abr{[a],b}\ge \abr{[a+b],b}$
\end{lemma}
\begin{proof}[Proof of \cref{lemma:rotationorderpreserving}]
Assume without loss of generality (by rotation) that $\arg(b) = 0$. Thus $b = (x,0)$ for $x \ge 0$. A point at $a = (u,v)$ is translated to $(u+x,v)$, thus the angle between the $x$-axis is smaller.
Hence $\abr{a,b}\ge \abr{a+b,b}$.

Now assume first that $\arg(a) \in [0,\pi]$.
Then $\arg(a) = \abr{a,b}$ and $\arg(a) \ge \abr{a+b,b} = \arg(a+b) \ge 0$.
From the fact that $b = [b]$ it follows that $\arg(b) = 0$ is a viable angle in our rounding.
As we assume minimal error rounding on the angle, it follows that $\arg([a]) \ge \arg([a + b])$ and hence $\abr{[a],b}\ge \abr{[a+b],b}$.

If $\arg(a) \in [- \pi, 0]$, we have $- \arg(a) = \abr{a,b}$ and 
$\arg(a) \le - \abr{a+b,b}  = \arg(a+b) \le 0$.
As before, we can conclude $\abr{[a],b} \ge \abr{[a+b],b}$.
\end{proof}

\angledecreasing*
\begin{proof}[Proof of \cref{corollary:phidecreasing}]
  We show that for all $i \geq N +1$:
  \[\abr{\lambda \vecitcomp{x}{i}{\ck{-}1}, \vecitcomp{x}{i}{\ck}} \ge \abr{\lambda \vecitcomp{x}{i+1}{\ck{-}1}, \vecitcomp{x}{i+1}{\ck}}\]
We first make the following calculating:
\begin{align}
\abr{[\lambda \vecitcomp{x}{i}{\ck{-}1}], \vecitcomp{x}{i}{\ck}} &\geq \abr{[\lambda \vecitcomp{x}{i}{\ck{-}1} + \vecitcomp{x}{i}{\ck}], \vecitcomp{x}{i}{\ck}} \tag{\Cref{lemma:rotationorderpreserving}}
\\&= \abr{\vecitcomp{x}{i+1}{\ck{-}1}, \vecitcomp{x}{i}{\ck}} \tag{by definition}
\\&= \abr{[\lambda \vecitcomp{x}{i+1}{\ck{-}1}],[\lambda \vecitcomp{x}{i}{\ck}]} \tag{1} \label{prop:1}
\\&= \abr{[\lambda \vecitcomp{x}{i+1}{\ck{-}1}],\vecitcomp{x}{i+1}{\ck}} \tag{$\ck{}$ is just rotating}
\end{align}
\cref{prop:1}: Since $\vecitcomp{x}{i+1}{\ck{-}1}$ and $\vecitcomp{x}{i}{\ck}$ are both at admissible angles, their rotations are at the same point between two admissible angles.
Thus the rotation-effect of the rounding will be the same for both values.

Observe that $[\lambda \vecitcomp{x}{i+1}{\ck{-}1}]$ and $\vecitcomp{x}{i+1}{\ck}$ both lie on admissible angles. Consider $\theta = \abr{[\lambda[a]],\lambda[a]}$, observe this angle (and the direction of the angle) is the same, no matter the value of $a$. This angle corresponds with $\abr{[\lambda \vecitcomp{x}{i+1}{\ck{-}1}],\lambda \vecitcomp{x}{i+1}{\ck{-}1}}$ and $\abr{[\lambda \vecitcomp{x}{i}{\ck{-}1}],\lambda \vecitcomp{x}{i}{\ck{-}1}}$. Further $\theta \le \theta_g/2$ the maximum effect of the rounding.

\begin{case}[{Suppose $\abr{[\lambda \vecitcomp{x}{i}{\ck{-}1}], \vecitcomp{x}{i}{\ck}} < \pi$}]
The calculation above also shows that $\arg([\lambda \vecitcomp{x}{i}{\ck{-}1}])$ relative to $\arg(\vecitcomp{x}{i}{\ck})$ is positive if and only if $\arg([\lambda \vecitcomp{x}{i+1}{\ck{-}1}])$ relative to $\arg(\vecitcomp{x}{i+1}{\ck{-}1})$ is positive.
That is, not only does the angle-distance decrease, but the relative position of the two points stays the same.

Because of this, and the fact that both $\vecitcomp{x}{i}{\ck{-}1}$ and $\vecitcomp{x}{i+1}{\ck{-}1}$ lie on admissible angles, the angle effect of rounding $\lambda \vecitcomp{x}{i}{\ck{-}1}$ relative to $\vecitcomp{x}{i}{\ck}$ is the same as the angle effect of rounding $\lambda \vecitcomp{x}{i+1}{\ck{-}1}$ relative to $\vecitcomp{x}{i+1}{\ck}$.
Hence we can conclude:
\[\abr{\lambda \vecitcomp{x}{i}{\ck{-}1}, \vecitcomp{x}{i}{\ck}} \ge \abr{\lambda \vecitcomp{x}{i+1}{\ck{-}1}, \vecitcomp{x}{i+1}{\ck}}\]
\end{case}
\begin{case}[{Suppose $\abr{[\lambda \vecitcomp{x}{i}{\ck{-}1}], \vecitcomp{x}{i}{\ck}} = \pi$ and $\abr{[\lambda \vecitcomp{x}{i+1}{\ck{-}1}],\vecitcomp{x}{i+1}{\ck}} < \pi$}]

Given $R\theta_g = \pi = \abr{[\lambda \vecitcomp{x}{i}{\ck{-}1}], \vecitcomp{x}{i}{\ck}} $ and  $\abr{[\lambda \vecitcomp{x}{i}{\ck{-}1}],\lambda \vecitcomp{x}{i}{\ck{-}1}}\le \theta_g/2$ we have 
\[
\phi(i) =  \abr{[\lambda \vecitcomp{x}{i}{\ck{-}1}],\vecitcomp{x}{i}{\ck}} - \abr{[\lambda \vecitcomp{x}{i}{\ck{-}1}],\lambda \vecitcomp{x}{i}{\ck{-}1}} \ge (R-1)\theta_g +\theta_g/2
\]

Given $\abr{[\lambda \vecitcomp{x}{i+1}{\ck{-}1}],\vecitcomp{x}{i+1}{\ck}} \le (R-1)\theta_g$ 
we have \[
\phi(i+1) \leq  \abr{[\lambda \vecitcomp{x}{i+1}{\ck{-}1}],\vecitcomp{x}{i+1}{\ck}} + \abr{[\lambda \vecitcomp{x}{i+1}{\ck{-}1}],\lambda \vecitcomp{x}{i+1}{\ck{-}1}} \le (R-1)\theta_g +\theta_g/2
\]
Hence $\phi(i) \ge (R-1)\theta_g +\theta_g/2 \ge \phi(i+1)$.
\end{case}
\begin{case}[{Suppose $\abr{[\lambda \vecitcomp{x}{i}{\ck{-}1}], \vecitcomp{x}{i}{\ck}} =\abr{[\lambda \vecitcomp{x}{i+1}{\ck{-}1}],\vecitcomp{x}{i+1}{\ck}} = \pi$}]

The rotation by $\abr{[\lambda \vecitcomp{x}{i}{\ck{-}1}], \lambda\vecitcomp{x}{i}{\ck{-}1}}$ in either direction results in the angle (when renormalised into $[0,\pi]$) of
\[
\phi(i) =  \abr{[\lambda \vecitcomp{x}{i}{\ck{-}1}],\vecitcomp{x}{i}{\ck}}-\abr{[\lambda \vecitcomp{x}{i}{\ck{-}1}],\lambda \vecitcomp{x}{i}{\ck{-}1}}
\]

Similarly, since the effect is the same at $i+1$ we have
\[
\phi(i+1) =  \abr{[\lambda \vecitcomp{x}{i+1}{\ck{-}1}],\vecitcomp{x}{i+1}{\ck}}-\abr{[\lambda \vecitcomp{x}{i+1}{\ck{-}1}],\lambda \vecitcomp{x}{i+1}{\ck{-}1}}
\]

Since $\abr{[\lambda \vecitcomp{x}{i}{\ck{-}1}],\vecitcomp{x}{i}{\ck}} = \abr{[\lambda \vecitcomp{x}{i+1}{\ck{-}1}],\vecitcomp{x}{i+1}{\ck}}$ and $\abr{[\lambda \vecitcomp{x}{i}{\ck{-}1}],\lambda \vecitcomp{x}{i}{\ck{-}1}} = \abr{[\lambda \vecitcomp{x}{i+1}{\ck{-}1}],\lambda \vecitcomp{x}{i+1}{\ck{-}1}} $ we have 
 $\phi(i) =  \phi(i+1)$.\qedhere
\end{case}
\end{proof}

\willperiod*
\begin{proof}[Proof of \cref{lemma:willperiod}]
  We show that for all $i \ge N$: $\vecitcomp{x}{i+1}{\ck{-}1} = [\lambda \, \vecitcomp{x}{i}{\ck{-}1}]$.
  First let $i = N$.
  We have
  \begin{align*}
    \abr{\lambda\vecitcomp{x}{N+1}{\ck{-}1}, \vecitcomp{x}{N+1}{\ck}}  &= \abr{\lambda\vecitcomp{x}{N}{\ck{-}1}, \vecitcomp{x}{N}{\ck}} \tag{by assumption}\\
    &= \abr{\lambda[\lambda\vecitcomp{x}{N}{\ck{-}1}], [\lambda \vecitcomp{x}{N}{\ck}]} \tag{1} \label{eq:1:willperiod}\\
    &= \abr{\lambda[\lambda\vecitcomp{x}{N}{\ck{-}1}], \vecitcomp{x}{N+1}{\ck}} \tag{$\ck{}$ is just rotating}
  \end{align*}
\Cref{eq:1:willperiod} holds because $\vecitcomp{x}{N}{\ck{-}1}$ and $\vecitcomp{x}{N}{\ck}$ are both at admissible points, and hence rounding after rotating by $\lambda$ has the same effect on both sides.

It follows that $\arg(\lambda\vecitcomp{x}{N+1}{\ck{-}1}) = \arg(\lambda[\lambda\vecitcomp{x}{N}{\ck{-}1}])$.
As $\abs{\vecitcomp{x}{N+1}{\ck{-}1}} = \abs{\vecitcomp{x}{N}{\ck{-}1}}$ by assumption, we can conclude that $\vecitcomp{x}{N+1}{\ck{-}1} = [\lambda \vecitcomp{x}{N}{\ck{-}1}]$.

  The next step from $N{+}1$ to $N{+}2$ is just a rotation of this case, and hence we can conclude by induction.
\end{proof}

\begin{lemma}\label{lemma:below90}
For $R\ge 3$: $\ceil{\frac{\pi/4}{\theta_g}}\theta_g + \frac{\theta_g}{2} \le \pi/2$, where $\theta_g= \frac{\pi}{R}$.
\end{lemma}
\begin{proof}[Proof of \cref{lemma:below90}]
$\ceil{\frac{\pi/4}{\theta_g}}\theta_g + \frac{\theta_g}{2} \to \pi/4$ as $\theta_g\to 0$ (or $R\to \infty$). Enumeration of the first 100 cases concludes less than $\pi/2$ before being close to $\pi/4$.
\end{proof}

\begin{figure}
\centering
\includegraphics[width=0.85\textwidth]{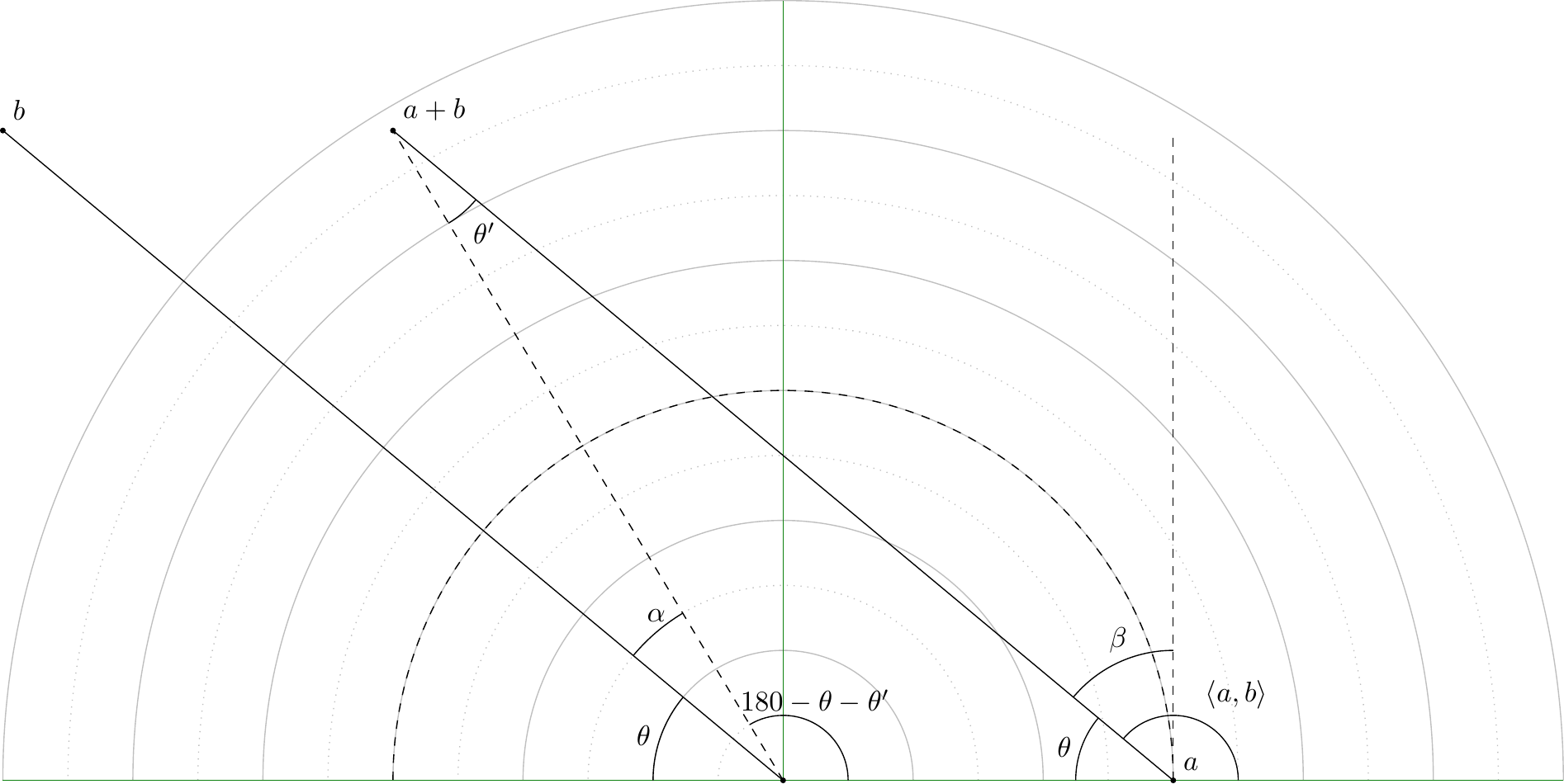}
\caption{Situation in the upper left quadrant ($\phi(i) > \pi/2$) after an increasing step. $\alpha \le \pi/4$, and so $\phi(i+1) \le \pi/2$ from this point on.}
\label{fig:sketch:lemma:ulq}
\end{figure}
\ulq*
\begin{proof}[Proof of \cref{lemma:ulq}] Refer to \cref{fig:sketch:lemma:ulq}, let $\alpha = \abr{a + b , b}$, $\beta = \abr{a,b} - \pi/2$ and $\theta = \pi/2 - \beta = \pi - \abr{a,b}$.

First we claim $\alpha \le \pi/4$: If $\theta > \pi/4$ and since $\alpha \le \pi/2 - \theta$, we have $\alpha \le \pi/4$. Otherwise suppose $\theta \le \pi/4$ and observe that $\theta' = \alpha$. Because  $\abs{a+b} > \abs{a}$ we have $\theta' < \theta$ so  $\alpha \le \pi/4$.

Suppose $\theta_g = \pi/2$, which implies $\theta_g/2 = \pi/4$ and $\alpha \le \theta_g/2$. Since $\abr{a + b , b} < \theta_g/2$ then $\abr{[ a + b] , b} = 0$  since $b$ is already rounded and $[\lambda a + b]$ is within $\theta_g/2$. then $\abr{\lambda[ a + b] , [\lambda b]} \le \theta_g/2$.

Instead suppose $\theta_g = \pi/R$, $R\ge 3$. Then $\abr{a + b , b} \le \pi/4$ and hence $\abr{[ a + b ], b} \le \ceil{\frac{\pi/4}{\theta_g}}\theta_g$ and $\abr{\lambda[ a + b] , [\lambda b]} \le  \ceil{\frac{\pi/4}{\theta_g}}\theta_g + \theta_g/2$, then by \cref{lemma:below90} $\ceil{\frac{\pi/4}{\theta_g}}\theta_g + \theta_g/2 \le \pi/2$.
\end{proof}

\begin{lemma}\label{lemma:urq}
Let $a_i = \lambda\vecitcomp{x}{i}{\ck{-}1}$ and $b_i = \vecitcomp{x}{i}{\ck}$.

Suppose $\abr{a_i,b_i} > \pi/2$ and $\abr{a_i+b_i,a_i} \le \pi/2$ and $\abs{a_{i+1}} = \abs{[a_i+b_i]} < \abs{a_i}$ (i.e. $\abs{\vecitcomp{x}{i+1}{\ck{-}1}} <  \abs{\vecitcomp{x}{i}{\ck{-}1}}$).

If $\abr{a_{i+1},b_{i+1}} = \abr{a_i,b_i}$ and $\abr{a_{i+1}+b_{i+1},a_{i+1}}\le \pi/2$ then $\abs{a_{i+2}}\le \abs{a_{i+1}}$ 
(entailing $\abs{\vecitcomp{x}{i+2}{\ck{-}1}}\le \abs{\vecitcomp{x}{i+1}{\ck{-}1}}$).
\end{lemma}

\begin{figure}
\centering
\includegraphics[width=0.85\textwidth]{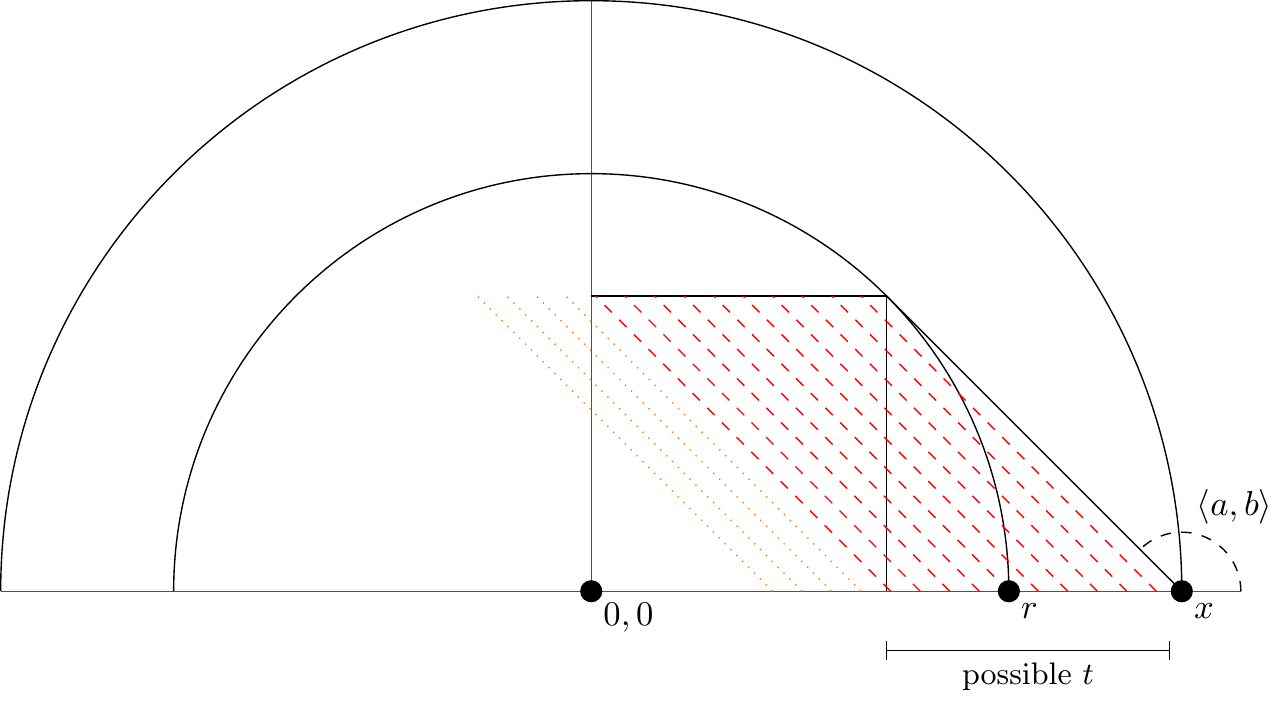}
\caption{Situation in the upper right quadrant after a decreasing step, indicating that if the next step stays inside the upper right quadrant, it cannot increase. Orange/dotted can be excluded as then $\abr{a+b,a} > \pi/2$.}
\label{fig:sketch:lemma:urq}
\end{figure}
\begin{proof}[Proof of \cref{lemma:urq}]
The situation is depicted in \cref{fig:sketch:lemma:urq}. By rotational symmetry, assume $a_{i} = (x,0)$. By applying the same rotational normalisation to $b_i$ let $b = (-l,m)$. Let $r,t$ be such that: $\abs{a_{i}+b} = r$ (inner black circle) and $\abs{[a_{i}+b]} = t < \abs{a_i} (=x)$.

Then by rotational symmetry again (using $\abr{a_{i+1},b_{i+1}} = \abr{a_i,b_i}$), assume $a_{i+1} = (t,0)$. Then observe that $\abs{(t,0) + b} \le r$ (red/dashed lines), and since $[r] = t$ and minimal error rounding is used $\abs{[(t,0) +b ]} \le t$.  We have $a_{i+2} = \lambda[a_{i+1}+b] = \lambda[(t,0)+b]$, hence $\abs{a_{i+2}} \le t$.
\end{proof}

\nodecreaseincrease*
\begin{proof}[Proof of \cref{corr:nodecreaseincrease}]
Recall $\gamma(i) = \abr{\lambda \vecitcomp{x}{i}{\ck{-}1} + \vecitcomp{x}{i}{\ck},\lambda \vecitcomp{x}{i}{\ck{-}1}}$.

Suppose following a modulus  decreasing transition (hence $\phi(i) \ge \pi/2$) there is  $\phi(i) = \phi(i+1) $ and $\abs{\vecitcomp{x}{i+1}{\ck{-}1}} > \abs{\vecitcomp{x}{i}{\ck{-}1}}$ 
\begin{itemize}
\item the last step was a decrease with $\gamma(i-1) \le \pi/2$ and $\gamma(i) \le \pi/2$ then by \cref{lemma:urq} it is not increasing (contradicting $\abs{\vecitcomp{x}{i+1}{\ck{-}1}} > \abs{\vecitcomp{x}{i}{\ck{-}1}}$).
\item the last step was a decrease with $\gamma(i-1) > \pi/2$ and this step has $\gamma(i) \le \pi/2,$ this cannot happen without a change of angle (contradicting $\phi(i-1)=\phi(i)$).
\item this step has $\gamma(i)> \pi/2$, then by \cref{lemma:ulq} an increase ($\abs{\vecitcomp{x}{i+1}{\ck{-}1}} > \abs{\vecitcomp{x}{i}{\ck{-}1}}$) would cause the angle to decrease (contradicting $\phi(i) = \phi(i+1) $).\qedhere
\end{itemize}
\end{proof}

\begin{lemma}\label{cor:notdecreasingagain:p1}
  Let $a = \lambda \vecitcomp{x}{N}{\ck{-}1}$ and $b = \vecitcomp{x}{N}{\ck}$ for some $N > 0$, and let $\gamma = \abr{a+b,a}$.
  Suppose that $\ck$ is just rotating after $N$, $\gamma < \pi/2$ and $\abs{a+b} > \abs{a} + 0.5$.

Then, for all $i > N$: $\abs{\vecitcomp{x}{i}{\ck{-}1}} > \abs{a}$.
\end{lemma}
\begin{proof}
  Recall that $\phi(i) = \abr{\lambda \vecitcomp{x}{i}{\ck{-}1},\vecitcomp{x}{i}{\ck}}$ and let $\phi = \phi(N)$.
  By \cref{corollary:phidecreasing} we have: $\phi(i) \leq \phi$ for all $i \geq N$.
  We show the claim by induction on $i$.
  If $i = N + 1$, it holds by assumption.
  Else, we can assume that $\abs{\vecitcomp{x}{i{-}1}{\ck{-}1}} > \abs{a}$.
  We let $d = \vecitcomp{x}{i{-}1}{\ck{-}1}$, $e = \vecitcomp{x}{i{-}1}{\ck}$ and we aim to show that $\abs{d + e} \geq \abs{a + b}$.
  \begin{alignat*}{2}
    & \abs{d + e} && \geq \abs{a + b} \\
    \iff& \abs{d + e}^2 &&\geq \abs{a + b}^2 \\
    \Longleftarrow \;\;& \abs{d}^2 + \abs{b}^2 - 2\abs{d} \abs{b} \cos(\pi {-} \phi) &&\geq \abs{a}^2 + \abs{b}^2 - 2\abs{a} \abs{b} \cos(\pi {-} \phi)  \tag{*} \\
    \iff& (\abs{a} + C)^2 - 2 (\abs{a} + C) \abs{b} \cos(\pi {-} \phi) &&\geq \abs{a}^2 - 2\abs{a} \abs{b} \cos(\pi {-} \phi) \tag{**} \\
    \iff& 2\abs{a}C + C^2 - 2 C \abs{b} \cos(\pi {-} \phi) && \geq 0 \\
    \iff& 2\abs{a}C + C^2 && \geq 2 C \abs{b} \cos(\pi {-} \phi) \\
    \Longleftarrow \;\;& \abs{a} &&\geq \abs{b}\cos(\pi {-} \phi)
  \end{alignat*}
  To see that $\abs{a} \geq \abs{b}\cos(\pi {-} \phi)$ holds, it is enough to observe that $\abs{a} = \abs{b}\cos(\pi {-} \phi) + \abs{a+b} \cos(\gamma)$ (see~\cref{fig:sketchpolar}), and $0 < \gamma < \frac{\pi}{2}$.
  The step $(*)$ is valid as $\ck$ is just rotating, which implies $\abs{b} = \abs{e}$, together with $\phi(i{-}1) \leq \phi$, which implies $\cos(\pi {-} \phi) \geq \cos(\pi {-} \phi(i{-}1))$.
  In $(**)$ we use that $\abs{d} = \abs{\vecitcomp{x}{i{-}1}{\ck{-}1}} = \abs{a} + C$, for some positive $C$.

  Now, from $\abs{d+e} \geq \abs{a + b} > \abs{a} + 0.5$ it follows that $\abs{[d+e]} = \abs{\vecitcomp{x}{i}{\ck{-}1}} > \abs{a}$.
\end{proof}

\begin{figure}
  \centering
  \includegraphics[width=0.7\textwidth]{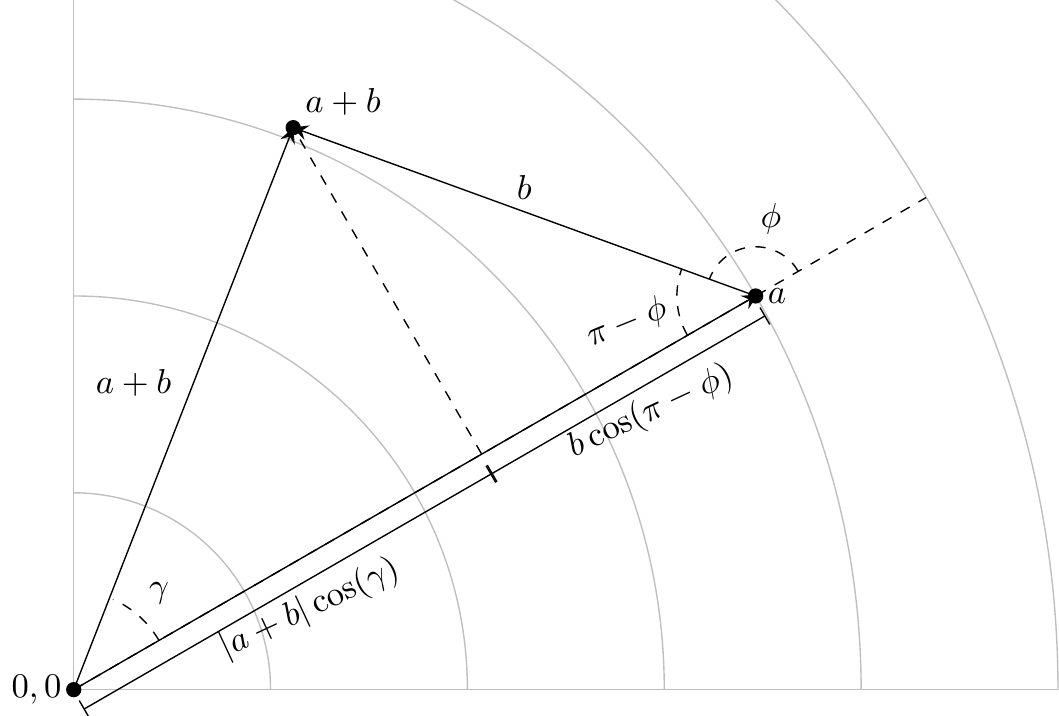}
  \caption{}
  \label{fig:sketchpolar}
\end{figure}

\begin{lemma}\label{cor:notdecreasingagain:p2}
  Let $a,b$ be algebraic numbers, and assume that $\abs{a+b} > \abs{a} \geq \frac{1}{\sqrt{2}}\abs{b}$.

  Then, $\abr{a+b,a} \leq \pi/2$.
\end{lemma}
\begin{proof}
  Let $\gamma = \abr{a+b,a}$ and assume, for contradiction, that $\pi/2 < \gamma \leq \pi$.
  We have:
  \begin{align*}
    \abs{b}^2 = \abs{a+b}^2 + \abs{a}^2 - 2\abs{a}\abs{a+b}\cos(\gamma) > 2\abs{a}^2 - 2\abs{a}\abs{a+b}\cos(\gamma) > 2\abs{a}^2
  \end{align*}
  The last step follows by $\cos(\gamma) < 0$.
  But then we have $\abs{b} > \sqrt{2} \abs{a}$, which contradicts the assumptions.
\end{proof}
\notdecreasingagain*
\begin{proof}[Proof of \cref{cor:notdecreasingagain}]
Direct corollary of \cref{cor:notdecreasingagain:p1} and \cref{cor:notdecreasingagain:p2}.
\end{proof}

\begin{proposition}\label{lemma:expspace}
  The problem (of \cref{thm:polar-expspace}) is in \EXPSPACE.
\end{proposition}

\begin{proof}[Proof of \cref{lemma:expspace}]
  This proof gives a more detailed analysis of the algorithm as presented in~\Cref{thm:polar-expspace}, showing that it only uses exponential space.
  As in~\Cref{thm:polar-expspace} we consider each dimension separately, starting with dimension $\dimension$.
  For each dimension ${\ck}$ we will establish upper bounds $T_{\ck}$ on the number of steps that we need while considering dimension ${\ck}$, and $U_{\ck}$ on the maximum numeric value in dimension $\ck$ that we need to consider.

  We assume that we are at step $N$ and have concluded that dimension $\ck$ just rotates at the right modulus (to conclude this for dimension $\dimension$ we need only one step).
  We use the same stopping criteria as in \cref{thm:polar-expspace} to conclude that we no longer reach $\target$.
  Observe that in \cref{thm:polar-expspace} the value of $\abs{\vecitcomp{x}{i}{\ck{-}1}}$ never increases beyond the value $K = \max(\abs{\target_{\ck}},\abs{\target_{\ck{-}1}},\abs{\vecitcomp{x}{N}{\ck{-}1}})$; thus if $\abs{\vecitcomp{x}{i}{\ck{-}1}} > K$  we can conclude that we can stop.
  This uses that the previous dimension is rotating on the ``right'' modulus, that is $\abs{\target_{\ck}} = \abs{\vecitcomp{x}{i}{\ck}}$ holds.
  
  The $[K]$-ball (see~\cref{def:kball}) has $K \cdot \frac{2\pi}{\theta_g}$ admissible points and hence, by the above argument, after this many steps we will either have: left the ball and concluded that we can stop, become just rotating in the current dimension, or decreased $\phi(i)$.
  As there are $\frac{2\pi}{\theta_g}$ possible values of $\phi$, in total we will spend at most $K \cdot (\frac{2\pi}{\theta_g})^2$ steps in dimension $\ck{-}1$.
 
  To ease calculations, we want to assume that the time spent in each dimension increases with respect to the previous one, and so we set $T_{\ck{-}1} = K \cdot (\frac{2\pi}{\theta_g})^2 + T_k$.
  This lets us assume that $N$, the step at which we start considering dimension $\ck{-}1$, satisfies: $N \leq \dimension \cdot T_{\ck}$.
We note that $\abs{\vecitcomp{x}{N}{\ck{-}1}}$ is at most $\abs{\vecitcomp{x}{0}{\ck{-}1}} + N \cdot U_{\ck}$, as $U_{\ck}$ is an upper bound on dimension $\ck$.
Hence we put $U_{\ck{-}1} = \abs{\vecitcomp{x}{0}{\ck{-}1}} + d \cdot T_k \cdot U_{\ck}$
  
  Using these upper bounds we now calculate how large the values $U_j$ may get with decreasing $j$ (we start at $j = \dimension$).
  In order not to distinguish the initial and target values of each dimension, we overestimate by using $\target_{s} = \sum_{j = 0}^d \abs{\target(j)}$ and $i_{s} = \sum_{j = 0}^d \abs{\vecitcomp{x}{0}{j}}$.

  \begin{align*}
    T_\dimension &= 1, \qquad \quad U_\dimension = i_{s} \\
    T_{\ck{-}1} &= K \cdot \left(\frac{2\pi}{\theta_g}\right)^2 + T_{\ck} \leq \left(\target_s + U_{\ck{-}1}\right) \cdot \left(\frac{2\pi}{\theta_g}\right)^2 + T_{\ck}\\[2mm]
    U_{\ck{-}1} &= i_s + \dimension \cdot T_{\ck} \cdot U_{\ck} \\
    &\leq i_s + \dimension \cdot \left(\left(\target_s + U_{\ck}\right) \cdot \left(\frac{2\pi}{\theta_g}\right)^2 + T_{\ck+1} \right) \cdot U_{\ck} \\
    &\leq \underbrace{i_s \cdot 3 \cdot \dimension \cdot \target_s \cdot \left(\frac{2\pi}{\theta_g}\right)^2}_{F}  \cdot U_{\ck}^2
  \end{align*}
  The last step uses that $T_{\ck + 1} \leq U_{\ck}$.
  Given the input, $F$ is fixed and of pseudopolynomial size.
  We can conclude that:
  \[U_{\dimension {-} j} \leq (F \cdot i_s)^{2^j}\]
  As $F \cdot i_s$ is exponential in the input, it follows that $U_0$ is at most double-exponential in the input.
  Hence, it requires at most exponentially many bits to express.
\end{proof}

\section{Additional material for Section~\ref{sec:truncated}, Argand truncation or expansion in Jordan normal form}

\textbf{In this subsection we assume all angles are given in degrees} as we will make use of rationality arguments on the angles. This is simply a stylistic choice, since it would be equivalent to consider rational multiples of $\pi$. 

\jnftruncation*
\begin{proof}[Proof of \cref{thm:jnf-tw-zero}]

Without loss of generality we consider only the case where $\abs{\lambda} = 1$, since if $\abs{\lambda} \ne 1$ the algorithm of \cref{thm:all-not-1} can be used on each such block in lock step.

Consider the $\dimension$th component. At each step, whenever rounding takes place, then there is some decrease in the modulus. Thus, either the coordinate hits zero (and stays forever), or it stabilises and becomes periodic (with no rounding ever occurring again). 

The $\dimension$th coordinate can be simulated until this happens. Then clearly it must match the target $y_\dimension$ occasionally, otherwise the answer is no. 

In the following we argue either it stabilises at zero in which case it is trivial. Or it becomes periodic with non-zero modulus and that this occurs if and only if $\arg(\lambda)$ is a multiple of $90$ degrees (and otherwise must reduce to zero).

\begin{case}[$\vecitcomp{x}{i}{\dimension}$ reaches zero] In this case it is stable, and the next coordinate is not effected by this coordinate (from some point on). Then the problem can be reduced to a smaller instance, by deleting the $\dimension$th coordinate.
\end{case}

\begin{case}[Periodic at modulus $>0$ and $\arg(\lambda)$ is not a root of unity] This case does not occur. If $\arg(\lambda)$ is not a root of unity then $\lambda^ix$ is dense on the circle of radius $\abs{x}$, and must eventually hit a point with non-integer coordinates. Such points must be rounded, decreasing the modulus, contradicting stability, and thus periodicity.

\end{case}
\begin{case}[Periodic at modulus $>0$ and $\arg(\lambda)$ is a root of unity]
If $\arg(\lambda)$ is a root of unity it is rational ($\lambda^n = 1$ implies $n \arg(\lambda) = k360$ for some $k \in \naturals{}$). We show the only angle that does not tend to zero is a multiple of $90$.

The following arguments assume we start at $a+bi$ and move to $c+di$ by a rotation of $\arg(\lambda)$, the proofs will be based on the rationality/irrationality of the angle and $\tan,\sin,\cos$ of the angle. To do this we assume both are in the upper right quadrant, as by rotating both by 90,180 or 270 to get it there will have the same argument regarding the irrationality.

Suppose we move from $a+0i$ to $c+di$. Recall the modulus is fixed, so $c^2+d^2 = a^2$. The angle formed by this is $\arg(\lambda)$ and the tangent is $\frac{d}{c}$. Since $c,d$ are rounded to a rational (but no rounding takes place) then the tangent is rational.
By \cref{thm:nivens}, the only point with $\arg(\lambda)$ rational and $\tan(\arg(\lambda))$ rational are $\arg(\lambda)= 45$ and $90$. Note that it cannot be $45$, because then $c+di = c+ci$, and we have $\sqrt{(c^2+c^2)} = a$. There is no integer solution to this equation (no Pythagorean triangle has angle 45 degrees). Hence to move from axis to non-axis the only acceptable angle is a multiple of 90 degrees (which indeed is not non-axis).

However, as a result of the finite period before stabilising and becoming periodic, the orbit could already be at a non-axis point, and move entirely within non-axis points. Suppose we move from $a+bi$ to $c+di$, with angle $\arg(\lambda)$. Note that $a+bi = C\exp(i\theta)$ and $c+di = C \exp(i(\theta + \arg(\lambda)))$ and hence $c+di = (a+bi)(\exp(i\arg(\lambda))) = (a+bi)(\cos(\arg(\lambda)) + i\sin(\arg(\lambda)))$. Then we have \[c = \operatorname{Re}(c+di) = \operatorname{Re}((a+bi)(\cos(\arg(\lambda)) + i\sin(\arg(\lambda)))) = a\cos(\arg(\lambda)) - b \sin(\arg(\lambda))\] and 
\[d = \operatorname{Im}(c+di) = \operatorname{Im}((a+bi)(\cos(\arg(\lambda)) + i\sin(\arg(\lambda)))) = b\cos(\arg(\lambda)) + a \sin(\arg(\lambda)).\]
Note then that 
\[
c + \frac{bd}{a} = (a+ \frac{b^2}{a}) \cos(\arg(\lambda)) \text{ and } d-\frac{bc}{a}  = (a + \frac{b^2}{a}) \sin(\arg(\lambda)),
\]
but since $a,b,c,d$ are rational we have $\cos(\arg(\lambda))$ and $\sin(\arg(\lambda))$ rational.
Recall, by \cref{thm:nivens}, the only point $\cos(\arg(\lambda))$ and $\arg(\lambda)$ are rational is $\arg(\lambda)$ multiple of $30$ (but not 60) or $90$ and the only point $\sin(\arg(\lambda))$ and $\arg(\lambda)$ are rational is  $\arg(\lambda)$ multiple of $60$ or $90$. Thus $\arg(\lambda)$ is a multiple of $90$.

In this case there is $\textit{no rounding}$ whatsoever. Indeed in this case the final coordinate, $x_\dimension$, is periodic after the first rounding step, with period at most $4$. If it starts at $a+bi$, it goes through (at most) $-b+ai,-a+bi,b-ai$ before returning back to $a+bi$. Then we show the penultimate coordinate, $x_{\dimension -1}$, grows from some point on giving a stopping criterion (either $\vecit{x}{i} = y$ at some point, or $\vecitcomp{x}{i}{\dimension-1} > y_{\dimension-1}$and never comes back). The initial point is invariant-under-rounding (i.e. when $x = [x])$, it is rotated by 90 degrees to another invariant-under-rounding point and then adds a point from the previous component (which is already invariant-under-rounding), resulting in an invariant-under-rounding point. Therefore we can use standard techniques to show that $\vecitcomp{x}{n}{\dimension-1} = \lambda \vecitcomp{x}{n-1}{\dimension-1} + \vecitcomp{x}{n-1}{\dimension}$ must grow;  To see this note that $\vecitcomp{x}{n}{\dimension-1} = \lambda^n \vecitcomp{x}{0}{\dimension-1} + n \lambda^{n-1}\vecitcomp{x}{0}{\dimension}$, which diverges as $n\to\infty$. Thus the analysis of components $1\dots \dimension -2$ is not necessary.

To see that this algorithm needs at most exponential space, we use a similar argument as in~\cref{lemma:expspace}.
First, we observe that the value in dimension $\dimension$ never increases.
Hence, an upper bound for the value in this dimension is $U_d = \abs{\vecitcomp{x}{0}{\dimension}}$.
This implies that we never exit the $[U_d]$-ball, and hence, after at most $T_d = \abs{[U_d]} \leq (2 U_d / g)^d$ steps we can conclude whether $\vecitcomp{x}{0}{\dimension}$ becomes periodic at $0$, or some other modulus.
In the latter case we must be in \textsf{Case 3}, and we can conclude that dimension $\dimension{-}1$ diverges.
Hence it is enough to simulate the system inside the $[\target_{\dimension{-}1}]$-ball, which is of single exponential size.

In the first case, we proceed to dimension $\dimension {-}1$, to which the same analysis applies, as now dimension $\dimension$ is at modulus 0 and does not influence the dynamics any more.
Dimension $\dimension{-}1$ may have grown to at most $U_{\dimension{-}1} = \abs{\vecitcomp{x}{0}{\dimension{-}1}} + T_d \cdot U_d$.

To simplify the calculations, we want to assume that $T_{\ck{-}1} > T_{\ck}$ and hence set: $T_{\ck{-}1} = \abs{[U_{\ck{-}1}]}  + T_{\ck} \leq (2 U_{\ck{-}1} / g)^d + T_{\ck}$.
This gives us $\sum_{j = \ck}^d T_{j} \leq \dimension \cdot T_{\ck}$.
We will use $\dimension \cdot T_{k}$ as an overestimate for the number of steps that were taken before reaching dimension $\ck{-}1$.
In order not having to distinguish the different initial values, we overestimate by assuming value $i_s = \sum_{j = 0}^\dimension \abs{\vecitcomp{x}{0}{j}}$ in every dimension.
Overall, this leads us to the following equations for $U_{\ck}$ and $T_{\ck}$:
  \begin{align}
     U_\dimension &= i_{s}, \qquad \quad T_\dimension = [U_{\dimension}] \leq (2U_{\dimension} / g)^d\\
     T_{\ck{-}1} &= [U_{\ck{-}1}] + T_{\ck} \leq \left(\frac{2 U_{\ck{-}1}}{g}\right)^d + T_{\ck}\\
    U_{\ck{-}1} &= i_s + \dimension \cdot T_{\ck} \cdot U_{\ck} \label{step:3} \\
    &\leq i_s + \dimension \cdot \left(\left(\frac{2 U_{\ck}}{g}\right)^\dimension + T_{\ck+1} \right) \cdot U_{\ck} \\
    &\leq \underbrace{i_s \cdot \dimension \cdot \left(2/g\right)^\dimension \cdot 2}_{F}(U_{\ck})^{d+1} \label{step:5}
  \end{align}
  Step~\ref{step:5} uses that $T_{\ck+1} < U_{\ck}$ (see Step~\ref{step:3}) and $d > 0$.
  Given the input, $F$ is fixed and single exponential.
  It follows that:
  \[U_{\dimension-j} \leq \left(F\cdot i_s\right)^{\left(d+1\right)^j}\]
  As $F \cdot i_s$ is single exponential, it follows that $U_0$ is at most double exponential in the input and hence expressible in single exponential space.
\qedhere
\end{case}
\end{proof}
 
\end{document}